\documentclass[11pt,a4paper]{amsart}
\usepackage[foot]{amsaddr}
\usepackage[T1]{fontenc}
\usepackage{lmodern}
\usepackage[top=1in, bottom=1in, left=1in, right=1in]{geometry}

\usepackage{natbib}
\usepackage{graphicx}
\usepackage{cancel}
\usepackage{caption}
\usepackage{pifont}
\usepackage{amsmath,amsfonts,amssymb,bbm,bm}
\usepackage{xspace}
\usepackage{enumitem}
\usepackage{amsthm}
\usepackage{booktabs} 
\usepackage{tikz-cd}
\usepackage{xcolor}
\usepackage{nameref}
\usepackage{multirow}
\usepackage{mathtools}
\usepackage{array}

\usepackage{tikz-3dplot} 
\usetikzlibrary{matrix,shapes,arrows,positioning,chains,patterns,calc, decorations.pathmorphing}
\DeclareMathOperator{\SW}{\tt EFFICIENCY}
\newcommand{\A}{\mathcal{A}}
\newcommand{\V}{\mathcal{V}}
\newcommand{\U}{\mathcal{U}}
\newcommand{\K}{\mathcal{K}}
\newcommand{\N}{\mathcal{N}}

\newcommand{\Sione}{\mathcal{S}_i^{1}}
\newcommand{\Sitwo}{\mathcal{S}_i^{2}}
\newcommand{\Sithw}{\mathcal{S}_i^{3}}
\newcommand{\Sifr}{\mathcal{S}_i^{4}}
\DeclareMathOperator{\OPT}{\textsc{OPT}\xspace}

\newcommand{\sw}{allocator's efficiency\xspace}
\newcommand{\ceil}[1]{\left\lceil#1\right\rceil}
\newcommand{\floor}[1]{\left\lfloor#1\right\rfloor}

\usepackage{hyperref} 
\hypersetup{colorlinks=true,citecolor=blue, linkcolor=blue, urlcolor=blue}
\usepackage{cleveref}
\usepackage{thmtools}
\usepackage{thm-restate}
\allowdisplaybreaks[4]

\RequirePackage[type1,tt=false]{libertine} \RequirePackage[libertine,vvarbb]{newtxmath} 
\RequirePackage[scr=rsfso]{mathalpha} 
\RequirePackage{bm}

\usepackage[linesnumbered,ruled,noend]{algorithm2e}
\usepackage{subfigure}
\usepackage{enumitem}
\SetKwInput{KwAssume}{Assumptions}
\SetKwInput{KwInput}{Input} 
\SetKwInput{KwOutput}{Output} 
\SetKwInput{KwPre}{Precondition}
\SetKwInput{KwInv}{Invariant}
\SetKwFunction{MainAlgorithm}{{\texttt MainAlgorithm}}
\SetKwFunction{BipartiteConstruction}{{\texttt BipartiteConstruction}} 

\newtheorem{theorem}{Theorem}[section]

\newtheorem*{claim*}{Claim}

\newtheorem{example}[theorem]{Example}

\newtheorem{lemma}[theorem]{Lemma}
\newtheorem{proposition}[theorem]{Proposition}
\newtheorem{problem}[theorem]{Problem}

\theoremstyle{definition}

\newtheorem{definition}[theorem]{Definition}
\newtheorem{remark}[theorem]{Remark}
\newtheorem*{remark*}{Remark}

\newtheorem{conjecture}[theorem]{Conjecture}

\newcommand{\abs}[1]{\left\vert#1\right\vert}

\title{\bf Fair Division with Allocator's Preference}
\thanks{A preliminary version appeared in Proceedings of the 19th Conference on Web and Internet Economics (WINE-2023).
Compared to the conference version, this version corrects an error in the proof of Theorem~\ref{thm:double_ef1}.
In the conference version, Theorem~\ref{thm:double_ef1} only holds under Conjecture~\ref{conj:kneser_conjecture}, and this version no longer depends on the conjecture.
The work was mostly done when Jiaxin was an undergraduate student at Shanghai Jiao Tong University.
}

\author{Xiaolin Bu, Zihao Li, Shengxin Liu, Jiaxin Song, Biaoshuai Tao}
\address[Xiaolin Bu, Biaoshuai Tao]{Shanghai Jiao Tong University. \textnormal{Email: \texttt{lin\_bu@sjtu.edu.cn}, \texttt{bstao@sjtu.edu.cn}.}}
\address[Zihao Li]{Nanyang Technological University. \textnormal{Email: \texttt{zihao004@e.ntu.edu.sg}.}}
\address[Shengxin Liu]{Harbin Institute of Technology, Shenzhen. \textnormal{Email: \texttt{sxliu@hit.edu.cn}.}}
\address[Jiaxin Song]{University of Illinois, Urbana-Champaign. \textnormal{Email: \texttt{jiaxins8@illinois.edu}.}
}

\date{}

\begin{document}
\begin{abstract}
We consider the problem of fairly allocating indivisible resources to agents. 
Most previous work focuses on fairness and/or efficiency \emph{among agents} given agents' preferences. 
However, besides the agents, the allocator, as the resource owner, may also be involved in many real-world scenarios (e.g., government resource allocation, heritage division, company personnel assignment, etc.). 
The allocator has the inclination to obtain a fair or efficient allocation based on her own preference over the items and to whom each item is allocated.
In this paper, we propose a new model and focus on the following two problems concerning the allocator's fairness and efficiency:
\begin{enumerate}
    \item Is it possible to find an allocation that is fair for both the agents and the allocator?
    \item What is the complexity of maximizing the \sw while ensuring agents' fairness?
\end{enumerate}

We consider the two fundamental fairness criteria: \emph{envy-freeness} and \emph{proportionality}.
For the first problem, we study the existence of an allocation that is envy-free up to $c$ goods (EF-$c$) or proportional up to $c$ goods (PROP-$c$) from both the agents' and the allocator's perspectives, in which such an allocation is called \emph{doubly EF-$c$} or \emph{doubly PROP-$c$} respectively. 
When the allocator's utility depends exclusively on the items (but not to whom an item is allocated), we prove that a doubly EF-$1$ allocation always exists. 
For the general setting where the allocator has a preference over the items \emph{and} to whom each item is allocated, we prove that a doubly EF-$1$ allocation always exists for two agents, a doubly PROP-$2$ allocation always exists for personalized bi-valued valuations, and a doubly PROP-$O(\log n)$ allocation always exists in general.

For the second problem, we provide various (in)approximability results in which the gaps between approximation and inapproximability ratios are asymptotically closed under most settings.
When agents' valuations are binary, the problems of maximizing the social welfare from the allocator's perspective while ensuring agents' fairness criteria of PROP-$c$ (with a general number of agents) and EF-$c$ (with a constant number of agents) are both polynomial-time solvable for any positive integer $c$.
For most of the other settings (general valuations, EF-$c$, etc.), we present strong inapproximability results. 
\end{abstract}

\maketitle

\setcounter{tocdepth}{1}
\tableofcontents

\section{Introduction}
\label{sec:intro}
Fair division studies how to fairly allocate a set of resources to a set of agents with heterogeneous preferences. It is becoming a valuable instrument in solving real-world problems, e.g., Course Match for course allocation at the Wharton School in the University of Pennsylvania~\citep{budish2017course}, and the website Spliddit (spliddit.org) for fair division of rent, goods, credit, and so on~\citep{goldman2015spliddit}.
The construct of fair division was first articulated by~\citet{Steinhaus48,Steinhaus49} in the 1940s, and has become an attractive topic of interest in a wide range of fields, such as mathematics, economics, computer science, and so on (see, e.g., \citep{brams1995envy,RobertsonWe98,procaccia2013cake,Moulin19,aziz2020developments,Suksompong21,AmanatidisAzBi23,Liu2023survey} for a survey).

The classic fair division problem mainly focuses on finding fair and/or efficient allocations \emph{among agents given agents' preferences}.
However, in many real-world scenarios, the allocator as the resource owner may also be involved, and, particularly, may have the inclination to obtain a fair or efficient allocation based on her own preference.
For example, consider the division of inheritances, e.g., multiple companies and multiple houses, from the parent to two children. Both children would prefer the companies as they believe the market value of the companies will be increased more than the houses in the future. At the same time, the parent may want to allocate the companies to the elder child since the parent thinks the elder child has a better ability to run the companies.
The final allocation should be fair for children and may also need to incorporate the parent's ideas about the allocation.
Another example is the government distributing educational resources (e.g., land, funding, experienced teachers or principals) among different schools. Some well-established schools may prefer land to build a new campus, while some new schools may need experienced teachers. On the other hand, the government may also have a preference (over the resources and to whom each resource is allocated) based on macroeconomic policy and may want the resulting distribution to be efficient, on top of each school feeling that it gets a fair share.
Other examples abound: a company allocates resources to multiple departments, an advisor allocates tasks/projects to students, a conference reviewer assignment system allocates papers to reviewers, etc.

We focus on the allocation of indivisible goods in this work.
To measure fairness, the two most fundamental criteria in the literature are \emph{envy-freeness} and \emph{proportionality}, respectively~\citep{Steinhaus48,Steinhaus49,Foley67,Varian74}.
In particular, an allocation is said to be envy-free if each agent weakly prefers her bundle over any other agent’s based on her own preference, and proportional if each agent values her bundle at least $1/n$ of her value for the whole resources, where $n$ is the number of agents.
Both fairness criteria can always be achieved in divisible resource allocation but it is not the case for indivisible resources (say, a simple example with two agents and one good).
This triggers an increasing number of research work to consider relaxing exact fairness notions of envy-freeness and proportionality to \emph{envy-freeness up to $c$ goods (EF-$c$)} and \emph{proportionality up to $c$ goods (PROP-$c$)} (see, e.g.,~\citep{Lipton04onapproximately,budish2011combinatorial,CFS17}). 
Specifically, an allocation is said to be EF-$c$ if any agent’s envy towards another agent could be eliminated by (hypothetically) removing at most $c$ goods in the latter's bundle, and PROP-$c$ if any agent's fair share of $1/n$ could be guaranteed by (hypothetically) adding at most $c$ goods that are allocated to other agents, where $c$ is a positive integer.
Besides fairness, another important issue of fair division is \emph{(economic) efficiency} (e.g., social welfare), which is used to measure the total happiness of the agents~\citep{cohler2011optimal,brams2012maxsum,10.5555/3306127.3331927,aziz2020computing}.

The fair division problem with allocator's preference presents new challenges compared to the classic fair division problems. 
With indivisible goods, it is well known that the \emph{round-robin algorithm}~\footnote{The round-robin algorithm works as follows: Given an ordering of agents, each agent picks her favorite item among the remaining items to her bundle following the ordering in rounds until there is no remaining item.} can return a fair, i.e., EF-1, allocation from the agents' perspective. 
However, this algorithm cannot be easily adapted to the problem where both agents and the allocator have preferences over items. Specifically, an agent's preference describes how much this agent values each item, while the allocator's preference describes how much the allocator regards each item values for each agent.
Consider the instance with both agents' and the allocator's preferences shown in Tables~\ref{tab:exp1} and~\ref{tab:exp2}.
\begin{figure}[h]
\begin{minipage}[c]{0.48 \textwidth}
\centering
\begin{tabular}{c|ccc}
    &  Item $1$ &  Item $2$ &  Item $3$ \\
    \hline
    Agent $1$ & $2$ & $1$ & $0$ \\ 
    Agent $2$ & $0$ & $1$ & $2$ 
\end{tabular}
\captionof{table}{Agents' Preferences}
\label{tab:exp1}
\end{minipage}
\begin{minipage}[c]{0.48\textwidth}
\centering
\begin{tabular}{c|ccc}
  &  Item $1$ &  Item $2$ &  Item $3$ \\
    \hline
    Agent $1$ & $0$ & $2$ & $1$ \\ 
    Agent $2$ & $1$ & $2$ & $0$ 
\end{tabular}
\captionof{table}{Allocator's Preferences}
\label{tab:exp2}
\end{minipage}
\end{figure}
Suppose, without loss of generality, agent 1 is before agent 2 in the ordering of the round-robin algorithm. When performing the algorithm without considering the allocator's preference, agent 1 gets a bundle of items 1 and 2, while agent 2 gets item 3.
From the allocator's perspective, this allocation is not EF-1 since the allocator thinks agent 2 will envy agent 1 even when an arbitrary item is removed from agent 1's bundle.
One can also verify that the above allocation is not social welfare maximizing based on the allocator's viewpoint, i.e., the allocator thinks there is another allocation such that the total happiness of the agents is larger. 
On the other hand, performing the round-robin algorithm based solely on the allocator's preference will return an allocation where agent 1 gets items 2 and 3 while agent 2 gets item 1 (assuming agent 1 has a higher priority in the ordering).
This allocation is not EF-1 from the agents' perspective, as the envy from agent 2 to agent 1 cannot be eliminated by removing a single item in agent 1's bundle.
This raises a natural question --  \emph{how to find fair or efficient allocations in the presence of agents' and the allocator's preferences?}

\subsection{Our Results}
We initiate the study of fair division with the allocator's preference and address the following two research questions in this paper:
\begin{enumerate}[leftmargin=*]
    \item[]{\bf Question 1}: \emph{how to find an allocation that looks fair to both the allocator and agents?}
    \item[]{\bf Question 2}: \emph{how to maximize the allocator's efficiency while ensuring agents' fairness?}
\end{enumerate}

We mainly focus on the allocation of \emph{indivisible} resources and discuss the \emph{divisible} resources in the appendix (which also includes omitted proofs in the paper).

\begin{table}[h]
    \centering
    \renewcommand{\arraystretch}{1.2} 
    \begin{tabular}{cc||>{\centering}p{1.8cm}>{\centering}p{1.8cm}cc}
        \hline
        valuation & $n$        & $v_i$ & $u_i$ & \textbf{doubly fairness} & \textbf{poly} \\
        \hline
        monotone  & 2   & arbitrary       & arbitrary           & EF-1 (Thm~\ref{thm:double_ef1}) & \ding{55} \\
        \hline
        \multirow{5}{*}{additive} & 2 & arbitrary & arbitrary &  EF-1 (Thm~\ref{thm:double_ef1_additive}) & \ding{51} \\
        \cline{3-6}
        & $2^k$ & arbitrary & arbitrary & PROP-$k$ (Thm~\ref{thm:double_logn}) & \ding{55} \\
        \cline{3-6}
        & \multirow{3}{*}{general} & arbitrary & identical & EF-1 (Thm~\ref{thm:identical_2ef1}) & \ding{51} \\
        & & \multicolumn{2}{c}{personalized bi-valued} & PROP-2 (Thm~\ref{thm:doubly_prop2_for_bivalued_utility}) & \ding{51}\\
        & & arbitrary & arbitrary & PROP-$(2\lceil\log n\rceil)$ (Thm~\ref{thm:double_logn}) & \ding{51}\\
        \hline 
    \end{tabular}
    \caption{Results of double fairness. The number of agents is denoted by $n$. For each agent $i$, $v_i$ represents her utility function while $u_i$ represents how much the allocator regards each item's value for agent $i$. 
    The column of ``poly'' means whether the proof is a polynomial-time constructive proof.
    }
    \label{tab:double_fair_result}
\end{table}

For the first problem, we propose new fairness notions \emph{doubly EF-$c$} and \emph{doubly PROP-$c$} that extend EF-$c$ and PROP-$c$ to our setting with regard to the allocator's preference.
Our results are presented in Table~\ref{tab:double_fair_result}.
For general monotone valuations, we show that a doubly EF-$1$ allocation always exists (Theorem~\ref{thm:double_ef1}).
Next, we consider additive valuations.
We first consider the setting where the allocator's utility only depends on the items (but not to whom an item is allocated), and we show that a doubly EF-$1$ allocation always exists and is polynomial-time computable (Theorem~\ref{thm:identical_2ef1}).
We then consider the general setting where the allocator's utility depends on both the items and the allocation.
For two agents, we show that a doubly EF-$1$ allocation always exists and is polynomial-time computable (Theorem~\ref{thm:double_ef1_additive}).
For a general number of agents, we show that a doubly PROP-$\log_2n$ allocation always exists for $n$ being an integer power of $2$, and a doubly PROP-($2\lceil\log n\rceil$) allocation always exists and can be computed in polynomial time (Theorem~\ref{thm:double_logn}).
If we restrict to personalized bi-valued valuations, we show that a doubly PROP-$2$ allocation always exists and can be computed in polynomial time (Theorem~\ref{thm:doubly_prop2_for_bivalued_utility}).

\begin{table}[t]
    \centering
    \renewcommand{\arraystretch}{1.2} 
    \begin{tabular}{cccc||cc}
        \hline
        $n$        & \textbf{Fairness} & $v_i$ & $u_i$ & \textbf{Negative Results} & \textbf{Positive Results} \\
        \hline
        \multirow{3}{*}{2}        & EF-$c$   & arbitrary       & arbitrary           & 2 (Thm~\ref{thm:mae_2_ab_neg})  & 2 (Thm~\ref{thm:mae_2_aa_pos}) \\
                 & EF-$c$   & arbitrary       & binary              & 2 (Thm~\ref{thm:mae_2_ab_neg})  & 2 (Thm~\ref{thm:mae_2_aa_pos}) \\
                 & EF-$c$   & binary          & arbitrary           & ---             & 1 (Thm~\ref{thm:mae_const_ba_pos}) \\
        \hline
        \multirow{2}{*}{constant} & EF-$c$   & arbitrary       & binary & $\floor{\frac{1+\sqrt{4n-3}}{2}}$~\citep{bu2022complexity} & unknown  \\
                 & EF-$c$   & binary          & arbitrary           & ---             & 1 (Thm~\ref{thm:mae_const_ba_pos}) \\
        \hline 
        \multirow{4}{*}{general}  & EF-$c$   & binary          & binary & $m^{1-\epsilon}, n^{1/2-\epsilon}$ (Thm~\ref{thm:mae_gen_bb_neg}) & $m$ (Thm~\ref{thm:mae_gen_aa_pos}) \\
                 & EF-$c$   & arbitrary       & arbitrary & $m^{1-\epsilon}, n^{1/2-\epsilon}$ (Thm~\ref{thm:mae_gen_bb_neg}) & $m$ (Thm~\ref{thm:mae_gen_aa_pos}) \\
                 & PROP-$c$ & arbitrary       & binary              & 2 (Thm~\ref{thm:maep_gen_ab_neg})   & unknown \\
                 & PROP-$c$ & binary          & arbitrary           & ---    & 1 (Thm~\ref{thm:maep_gen_ba_pos}) \\
        \hline 
    \end{tabular}
    \caption{Positive and negative results of maximizing the allocator's efficiency. The numbers of agents and items are denoted by $n$ and $m$, respectively. For each agent $i$, $v_i$ represents her utility function while $u_i$ represents how much the allocator regards each item values for agent $i$. Numbers $\alpha$ for negative results indicate that the problem is NP-hard to approximate to within the ratio $\alpha$; numbers $\alpha$ for positive results indicate that the problem admits a polynomial time $\alpha$-approximation algorithm. All our negative results hold for $c=1$.}
    \label{tab:mae_result}
\end{table}

For the second problem, we consider the complexity and approximability for both binary and general additive valuations.
Our results are presented in Table~\ref{tab:mae_result}.
The gap between the approximation ratio and the inapproximability ratio is closed, or asymptotically closed, under most settings.
If agents' valuations are binary, this problem is tractable for EF-$c$ with a constant number of agents and for PROP-$c$ with a general number of agents.
Under most other settings, this problem admits strong inapproximability ratios even for $c=1$.

Our results use several novel technical tools that are uncommon in the existing fair division literature, including i) the chromatic numbers of graphs as well as a topological proof, and ii) linear programming-based analyses.
\begin{itemize}[leftmargin=*]
    \item For i), we use the generalized Kneser graph (in \Cref{def:generalize_kneser_graph}) and $\Gamma$-graph (in \Cref{def:gamma_n}) to model the set of allocations and the relations between them.
    Specifically, the set of allocations that are not fair based on an agent's valuation forms an independent set in the graph.
    The existence of a doubly fair allocation is built upon the fact that, after removing all the independent sets --- each corresponding to a set of unfair allocations for an agent --- there are still remaining vertices, which is related to the chromatic number of the graph.
    In Sect.~\ref{sect:proof_chromatic_gamma_n}, we establish a lower bound of the chromatic number of the $\Gamma$-graph based on a topological analysis of the connectivity of the neighborhood complex $\mathcal{N}$, which corresponds to the $\Gamma$-graph.
    To demonstrate the connectivity,
    we show that the inclusion map for the 2-dimensional skeleton of  $\mathcal{N}$ is a nullhomotopic.
    \item For ii), we formulate our problems as linear programs.
    The solution to the linear program naturally corresponds to a \emph{fractional} allocation.
    Our technique is mainly based on the analysis of the vertices of the polytope defined by the linear program.
    In some applications, we handle the fractional items by rounding.
    In others, we prove that all the vertex solutions of the linear program are integral.
\end{itemize}

\subsection{Further Related Work}
\label{sect:relatedwork}

Conceptually, our model with allocator's preference shares similarities with recent research work on fair division with two-sided fairness, e.g., \citep{patro2020fairrec,gollapudi2020almost,FreemanMiSh21,IgarashiKaSu23}.
The existing two-sided fairness literature studies the fair division problem where there are two disjoint groups of agents and each agent in one group has a preference over the agents of the other group. The objective is then to find a (many-to-many) matching that is fair to each agent with respect to her belonging group.
We remark that these two models are different due to the following major reasons:
\begin{itemize}[leftmargin=*]
\item In their model, there are two disjoint sets of agents, and each group of preferences is defined from one set of agents to the other set of agents (viewed as a set of ``goods"). On the other hand, the two groups of preferences (one is from the agents and the other one is from the allocator) in our setting are both defined on a single set of agents and a single set of goods.

\item In their model, each agent will be allocated (or matched) a set of agents from the other group which is different from ours, whereas the allocator in our model will not receive any resource in the allocation.
\end{itemize}

As we can see, our model with allocator's preference reduces to the standard setting of indivisible goods when the allocator's preference coincides with agents' preferences.
Our first research question reduces to find EF-$c$ or PROP-$c$ allocations in indivisible fair allocation, where the fairness notions of EF-1 and PROP-1 are extensively studied.
In particular, an EF-1 allocation always exits and can be computed in polynomial time~\citep{Lipton04onapproximately,CKMP+19}.
For PROP-1, an allocation that is PROP1 and Pareto optimal always exits and can be computed in polynomial time~\citep{CFS17,barman2019proximity,aziz2020polynomial,mcglaughlin2020improving}. 
When considering the issue of economic efficiency, the problem in our second research question could be mapped to the problem of maximizing social welfare within either EF-1 or PROP-1 allocations in the indivisible goods setting. \citet{aziz2020computing} showed that the problem with either the EF-1 or the PROP-1 condition is NP-hard for $n \ge 2$ and \citet{10.5555/3306127.3331927} showed that the problem with the EF-1 requirement is NP-hard to approximate to within a factor of $1/m^{1-\varepsilon}$ for any $\varepsilon>0$ for general numbers of agents $n$ and items $m$.
Later, \citet{bu2022complexity} gave a complete landscape for the approximability of the problem with the EF-1 criterion.

In addition, several works studied the fair division problem where the resources need to be allocated among \emph{groups} of agents, and the resources are shared among the agents within each predefined group~\citep{manurangsi2017asymptotic,segal2019fair,suksompong2018approximate,segal2019democratic}.
In their model, $n=n_1+\cdots+n_k$ agents will be divided into $k \geq 2$ groups, where group $i$ contains $n_i \geq 1$ agents. An allocation is a partition of goods into $k$ groups. Each agent in the $i$-th group extracts utilities according to the $i$-th bundle. 
\citet{kyropoulou2020almost} also generalized the classic EF-$c$ to the group setting: An agent's envy towards another group could be eliminated by removing at most $c$ goods from that group's bundle.
PROP-$c$ could be defined similarly \citep{manurangsi2022almost}. 
With binary valuations, \citet{kyropoulou2020almost} gave the characterization of the cardinalities of the groups for which a group EF-1 allocation always exits.
In particular, they showed that a group EF-1 allocation always exists when there are two groups and each group contains two agents with binary valuations.
Subsequently, \citet{manurangsi2022almost} showed via the discrepancy theory that EF-$O(\sqrt{n})$ and PROP-$O(\sqrt{n})$ allocations always exist in the group setting.
Note that, when each group contains exactly two agents, i.e., $n_1 = \ldots = n_k =2$, the fair division problem in the predefined group setting coincides with our model (where each group could be considered to have an agent and the allocator).
However, we obtain improved results in this particular setting through different technical tools.

\section{Preliminaries}

Let $[k] = \{1,\ldots,k\}$. Our model consists of a set of agents $N = [n]$, a set of indivisible items $M = \{g_1, \ldots, g_m\}$, and \emph{the allocator}.
Each agent $i$ has a non-negative \emph{utility function} $v_i:\{0, 1\}^m \to \mathbb{R}_{\geq 0}$. 
In addition, the allocator has her preference, which is modeled by $n$ utility functions $u_i:\{0, 1\}^m \to \mathbb{R}_{\geq 0}$ where $u_i$ characterizes how much the allocator regards each item's value for agent $i$.

A utility function $v_i$ (resp., $u_i$) is said to be \emph{additive} if $v_i(X) = \sum_{g\in X} v_i(g)$ (resp., $u_i(X) = \sum_{g\in X} u_i(g)$) for any bundle $X\subseteq M$. 
A utility function $v_i$ (resp., $u_i$) is said to be \emph{monotone} if $v_i(S)\le v_i(T)$ (resp., $u_i(S)\le u_i(T)$) for any two bundles $S\subseteq T$.
We always assume $v_i(\emptyset)=0$.
It is clear that monotone valuations are much more general than additive ones.
We assume the utility functions to be additive throughout the paper, except for Sect.~\ref{sect:generalMonotone} where we consider monotone valuations beyond additive.

We also consider the following special cases of additive valuations.
A utility function $v_i$ (resp., $u_i$) is said to be \emph{binary} if $v_i(g) \in \{0,1\}$ (resp., $u_i(g) \in \{0,1\}$) for any item $g\in M$,
and is said to be \emph{personalized bi-valued} if $v_i(g) \in \{p_{i,v}, q_{i,v}\}$ (resp., $u_i(g) \in \{p_{i,u}, q_{i,u}\}$) with $0 \le p_{i,v} < q_{i,v} \le 1$ (resp., $0 \le p_{i,u} < q_{i,u} \le 1$) for any item $g \in M$.
It is obvious that the binary utility function is a special case of the personalized bi-valued utility function where $p_{i,v} = p_{i,u} = 0$ and $q_{i,v} = q_{i,u} = 1$.

An \emph{allocation} of the items $\A = (A_1, \dots, A_n)$ is an ordered partition of $M$, where $A_i$ is the bundle of items allocated to agent $i$. 
Below we introduce the fairness notions. 
Let $c$ be a non-negative integer. 

\begin{definition}[Envy-free up to $c$ goods]
An allocation $\mathcal{A}$ is said to be \textit{envy-free up to $c$ goods (EF-$c$)} if for all pairs of agents $i\neq j$, there exists a set $B\subseteq A_j$ such that $\abs{B} \le c$ and $v_i(A_i)\ge v_i\left(A_j\setminus B\right)$ (or $v_i(A_i)\geq v_i(A_j)-v_i(B)$ for additive utility functions). 
\end{definition}
\begin{definition}[Proportional up to $c$ goods]
An allocation $\mathcal{A}$ is said to be \textit{proportional up to $c$ goods (PROP-$c$)} if for any agent $i$, there exists a set $B \subseteq M\setminus A_i$ such that $\abs{B} \le c$ and $v_i(A_i\cup B)\geq v_i(M)/n$ (or $v_i(A_i)\ge \frac{1}{n} v_i(M) - v_i(B)$ for additive utility functions).
\end{definition}

Clearly, EF-$c$ implies PROP-$c$ for additive utility functions.
It is also well known that an EF-1 (hence, PROP-1) allocation always exists and can be computed in polynomial time~\citep{Lipton04onapproximately,CKMP+19}.
In our model, in addition to ensuring fairness among agents, we also consider the allocator's fairness. Thus, we generalize the above fairness criteria as follows.
\begin{definition}[Doubly envy-free up to $c$ goods]
An allocation $\mathcal{A}$ is said to be \textit{doubly envy-free up to $c$ goods (doubly EF-$c$)} if for all pairs of agents $i\neq j$, there exist sets $B_1, B_2\subseteq A_j$ such that $\abs{B_1}, \abs{B_2} \le c$, $v_i(A_i)\ge v_i(A_j\setminus B_1)$ and $u_i(A_i)\ge u_i(A_j\setminus B_2)$. 
\end{definition}
\begin{definition}[Doubly proportional up to $c$ goods]
An allocation $\mathcal{A}$ is said to be \textit{doubly proportional up to $c$ goods (doubly PROP-$c$)} if for any $i\in N$, there exist sets $B_1, B_2 \subseteq M\setminus A_i$ such that $\abs{B_1}, \abs{B_2}\le c$, and $v_i(A_i\cup B_1)\geq v_i(M)/n$, and $u_i(A_i\cup B_2)\geq u_i(M)/n$.
\end{definition}
When the allocator's utility functions are identical to agents' utility functions, it is easy to see that doubly EF-$c$ and doubly PROP-$c$ degenerate to EF-$c$ and PROP-$c$, respectively. 
The above-defined double fairness notions with the allocator's preference can also be interpreted as: there are two groups of valuation functions $u$ and $v$, where one is from the agents and the other one is from the allocator. A single allocation is said to satisfy double fairness if such an allocation is fair, e.g., doubly EF-$c$/PROP-$c$, with respect to both valuation functions $u$ and $v$.

To measure the economic efficiency of the allocator, we consider \emph{\sw}:
\begin{definition}\label{def:SW}
The \emph{\sw} of an allocation $\A=(A_1,\ldots,A_n)$, denoted by $\SW(\A)$, is the summation of the allocator's utilities of all the agents
$\SW(\A)=\sum_{i=1}^nu_i(A_i)$.
\end{definition}

In this paper, we are interested in the following two problems.
\begin{problem}
Given a set of indivisible items $M$, a set of agents $N=[n]$ with their utility functions $(v_1,\ldots,v_n)$, and the allocator with her preference $(u_1,\ldots,u_n)$, determine whether there exists an allocation $\A=(A_1,\ldots,A_n)$ that is doubly EF-$c$/PROP-$c$.
\end{problem}

\begin{problem}
Given a set of indivisible items $M$, a set of agents $N=[n]$ with their utility functions $(v_1,\ldots,v_n)$, and the allocator with her preference $(u_1,\ldots,u_n)$, the problem of \emph{maximizing \sw subject to EF-$c$/PROP-$c$} aims to find an allocation $\A=(A_1,\ldots,A_n)$ that maximizes \sw $\SW$ subject to that $\A$ is EF-$c$/PROP-$c$.
\end{problem}

\subsection{Totally Unimodular Matrix and Linear Programming}

\begin{definition}[Totally unimodular matrix]
    A matrix $\mathbf{A}_{m\times n}$ is a \emph{totally unimodular matrix} (TUM) if every square submatrix of $\mathbf{A}$ has determinant $0$, $+1$ or $-1$.
\end{definition}

To determine whether a matrix is TUM, we have the following lemma:
\begin{lemma}\label{lem:bipartite_utm}
Given a matrix $\mathbf{A} \in \{0, \pm 1\}^{m\times n}$, $\mathbf{A}$ is TUM if it can be written as the form of 
$\left[\begin{array}{c}
     \mathbf{A}_1  \\
     \mathbf{A}_2 
\end{array}\right]$, 
where $\mathbf{A}_1\in \{0, 1\}^{r\times n}$ (or $\{0, -1\}^{r\times n}$), $\mathbf{A}_2\in \{0, 1\}^{(m-r)\times n}$ (or $\{0, -1\}^{(m-r)\times n}$), $1\le r\le m$ and there is at most one nonzero number in every column of $\mathbf{A}_1$ or $\mathbf{A}_2$. 
\end{lemma}
\begin{proof}
We prove it by induction. Assume the square submatrix $\mathbf{A}'$ of $\mathbf{A}$ is an $n'\times n'$ matrix.
It holds when $n'=1$ since all entries are $0$, $-1$ or $1$.
We next assume $n'>1$. 
If there exists one column with only one non-zero entry, and we assume the square submatrix after removing the corresponding row and column of this entry is $\mathbf{B}$, we have $\det(\mathbf{A})=\pm \det(\mathbf{B})$. By induction, $\det(\mathbf{A}')$ is also equal to $0$, $-1$ or $1$, which concludes the lemma.

If there is no such column, since there are at most two non-zero entries in one column of the original matrix $\mathbf{A}$, there are exactly two entries in each column. 
Then, we consider the following linear combination of rows in $\mathbf{A}'$.
If $\mathbf{A}_1$ and $\mathbf{A}_2$ consist of the same non-zero values, we add all rows in $\mathbf{A}_1$ and minus all rows in $\mathbf{A}_2$.
Otherwise, we add all rows in both $\mathbf{A}_1$ and $\mathbf{A}_2$.
The above linear combination is equal to a zero vector, which implies $\det(\mathbf{A}')=0$.
Thus, $\mathbf{A}$ is totally unimodular.
\end{proof}

\begin{lemma}[\citet{hoffman2010integral}]
\label{lem.tum}
If $\mathbf{A}$ is totally unimodular and $\mathbf{b}$ is an integer vector, then each vertex of the polytope $\{\mathbf{A}\mathbf{x}\leq \mathbf{b},\mathbf{x}\geq\mathbf{0}\}$ has integer coordinates.
\end{lemma}

We can further show there exist polynomial-time algorithms to find the optimal vertex solution for such a linear program by the following lemma.

\begin{lemma}[\citet{guler1993degeneracy}]
\label{lem.lpvsol}
    For a linear program $\max\{\mathbf{c}^\top\mathbf{x}:\mathbf{A}\mathbf{x}\leq \mathbf{b},\mathbf{x}\geq\mathbf{0}\}$, if optimal solutions exist, an optimal vertex solution can be found in polynomial time.
    In particular, we can find an (integral) vertex of the polytope $\{\mathbf{A}\mathbf{x}\leq \mathbf{b},\mathbf{x}\geq\mathbf{0}\}$ in polynomial time.
\end{lemma}

\subsection{Kneser Graph and Chromatic Number}
Let $n$ and $k$ be two integers.
The \emph{Kneser graph} $\K(n, k)$ is a graph whose vertices are all subsets with $k$ elements of $[n]$, and two vertices are adjacent if their intersection is empty.
It was further extended to the following generalized version.

\begin{definition}
\label{def:generalize_kneser_graph}
    Given positive integers $n, k$ and a non-negative integer $s$, in the generalized Kneser graph $\K(n, k, s)$, the vertices are all the subsets with $k$ elements of $[n]$, and there is an edge between two vertices if and only if the corresponding two subsets intersect at most $s$ elements.
\end{definition}

The \emph{chromatic number} of a graph is the minimum number of colors needed to color the vertices such that no two adjacent vertices are colored the same. 
In other words, the vertices with the same color form an independent set.
Denote by $\chi(G)$ the chromatic number of a graph $G$.
Denote the chromatic number of a generalized Kneser graph $\K(n, k, s)$ by $\chi(n, k, s)$.  
For instance, when $n=4, k=3, s=2$, the generalized Kneser graph has $\binom{4}{3}=4$ vertices and every two vertices are adjacent. 
Thus, $\K(4,3,2)$ is a clique and $\chi(4,3,2) = 4$. 
When $n \ge 2k$, the chromatic number of the Kenser graph $\K(n, k)$ is equal to $n -2k +2$~\citep{LOVASZ1978319,Gre02,Bar78,10.1007/s00493-004-0011-1}. 
For the generalized Kneser graph, \citet{JAFARI2020111682} showed the following lower bound.
\begin{lemma}[\citet{JAFARI2020111682}]\label{lem:chromatic_for_kneser}
For every $0\le s < k < n$, $\chi(n, k, s) \ge n - 2k + 2s + 2$. 
\end{lemma}

\subsection{Lovasz's Bound and Fundamentals of Topology}
\label{sec:basics_topology}
In this part, we provide some basics of topology.
A \emph{$n$-simplex} is a geometric object with $(n+1)$ vertices in an $n$-dimension space.
In particular, a vertex is a $0$-dimensional simplex.
We use the words \emph{vertex}, \emph{edge}, \emph{triangle}, and \emph{tetrahedron} to refer to simplices with dimensions $0$, $1$, $2$, and $3$ respectively.

\begin{definition}[Simplicial complex]
A \emph{simplicial complex} $\K$ is a set of simplices that satisfies 1) a face of a simplex in $\K$ is also a simplex in $\K$, and 2) any non-empty intersection of every two simplices $\sigma_1, \sigma_2\in \K$ is a face of both simplices.    
\end{definition}
A simplicial complex can be viewed as a hereditary set system $\mathcal{F}$ of a ground set $[n]$ where there are $n$ vertices representing the $n$ elements and a simplex corresponds to a subset in $\mathcal{F}$.
A simplicial complex can also be viewed as a topological space.
For example, it can be embedded into an Euclidean space.
For a simplicial complex $\K$, we use $\tilde{\K}$ to denote its geometrical realization.

\begin{definition}[Neighborhood complex]
Given a graph $G$, its \emph{neighborhood complex} $\N(G)$ is a simplicial complex where its vertices are the vertices of $G$ and each of its simplices represents a subset of vertices of $G$ that share a common neighbor.
That is,
$$
\N(G) = \{\sigma \subseteq V: \exists v\in V, (a, v)\in E, \forall a\in \sigma \}\,.
$$
\end{definition}

\begin{definition}[Homotopic and null-homotopic]
A continuous mapping $f:X\rightarrow Y$ is said to be \emph{homotopic} to another continuous mapping $g:X\rightarrow Y$ if there exists a continuous mapping $F: X\times [0,1]$ such that $F(x, 0) = f(x)$ and $F(x, 1) = g(x)$. 
A continuous mapping $f:X\rightarrow Y$ is said to be \emph{null-homotopic} if it is \emph{homotopic} to a constant mapping.    
\end{definition}
\begin{definition}[$k$-connectivity]
\label{def:k_connect}
A topological space $T$ is called \emph{$k$-connected} if each continuous mapping of $S^r$ (the surface of the $(r+1)$-dimensional unit ball in $\mathbb{R}^{r+1}$) into $T$ extends continuously to the whole ball, for each $r=0,1,\ldots,k$.
Equivalently, every continuous mapping $f:S^r\to T$ is null-homotopic.
\end{definition}

Our proof will use the following celebrated result due to~\citet{LOVASZ1978319}.
\begin{restatable}[\citet{LOVASZ1978319}]{theorem}{Lovaszbound}
\label{thm:neighbor_chrom}
If $\tilde{\N}(G)$ is $k$-connected, then $\chi(G)\geq k+3$.
\end{restatable}

\section{Double Fairness of Two Agents with General Monotone Utilities}
\label{sect:generalMonotone}
In this section, we deal with general monotone utility functions and prove that a doubly EF-$1$ allocation always exists for two agents.

\begin{theorem}\label{thm:double_ef1}
When $n = 2$ and the valuations are monotone, a doubly EF-$1$ allocation always exists. 
\end{theorem}

To prove \Cref{thm:double_ef1}, we first show \Cref{prop:ef1_intersection}. Recall that $v_1$ and $v_2$ represent the utility functions for agents $1$ and $2$, respectively, and $u_1$ and $u_2$ correspond to the two valuation functions for the allocator.
The intuition behind this proposition is simple. 
If $|A_1\cup A_1'|\geq m-1$, then $A_1'$ contains all but at most one item in the complement of $A_1$. 
If the value of $A_1$ is too small to satisfy EF-$1$ for $v_1$ or $u_1$, then $A_1'$ must have a large enough value.
In addition, notice that $|A_1\cap A_1'|\leq 1$ implies $|A_2\cup A_2'|\geq m-1$.

\begin{restatable}{proposition}{EFOneIntersect}
\label{prop:ef1_intersection}
    Consider two allocations $(A_1,A_2)$ and $(A_1',A_2')$. If $|A_1\cup A_1'|\geq m-1$, then at least one of the two allocations is EF-$1$ with respect to $v_1$. The same holds for $u_1$. If $|A_1\cap A_1'|\leq 1$, then at least one of the two allocations is EF-$1$ with respect to $v_2$. The same holds for $u_2$.
\end{restatable}
\begin{proof}
    For the first part, suppose $|A_1\cup A_1'|\geq m-1$.
    Suppose without loss of generality that $(A_1,A_2)$ is not EF-$1$ for $v_1$.
    If $|A_1\cup A_1'|=m$, i.e., $A_1\cup A_1'=M$, we have $A_2=M\setminus A_1\subseteq A_1'$ and, as a result, $A_2'=M\setminus A_1'\subseteq M\setminus A_2=A_1$.
    Thus, $v_1(A_1')\geq v_1(A_2)>v_1(A_1)\geq v_1(A_2')$, where the middle inequality is due to that $(A_1,A_2)$ fails envy-freeness with respect to $v_1$ (since we have assumed the allocation is not even EF-$1$).
    This implies $(A_1',A_2')$ is envy-free with respect to $v_1$ and thus EF-$1$.
    If $|A_1\cup A_1'|=m-1$, let $g$ be the (only) item in the set $M\setminus(A_1\cup A_1')$.
    We know $g\in A_2$ and $g\in A_2'$.
    Consider the two allocations $(A_1,A_2\setminus\{g\})$ and $(A_1',A_2'\setminus\{g\})$ of the item set $M\setminus\{g\}$.
    We have $A_1\cup A_1'=M\setminus\{g\}$ and $(A_1,A_2\setminus\{g\})$ is not envy-free for $v_1$ (since $(A_1,A_2)$ is not EF-$1$).
    By the same analysis above, $(A_1',A_2'\setminus\{g\})$ must be envy-free for $v_1$, and $(A_1',A_2')$ must satisfy EF-$1$.
    The same arguments hold for $u_1$.

    For the second part, noticing that $|A_1\cap A_1'|\leq 1$ implies $|A_2\cup A_2'|\geq m-1$, the same arguments in the first part can be applied.
\end{proof}

Based on the above proposition, for any two allocations $(A_1, A_2)$ and $(A_1', A_2')$, if both of them are not EF-1 to some valuation (say $v_1$), then it should hold that $\abs{A_1\cup A_1'} < m-1$.
Consider the first bundles of all the allocations.
Next, we introduce an auxiliary graph, the $\Gamma$-graph, of which the vertex set consists of all the bundles and the edges are constructed based on \Cref{prop:ef1_intersection}.

\begin{definition}[$\Gamma$-Graph]
\label{def:gamma_n}
Given a positive integer $n$, let $\Gamma(n)$ be the graph with $2^n$ vertices corresponding to the $2^n$ subsets of $[n]$ and there is an edge between two vertices if the corresponding two subsets $A$ and $B$ satisfy $|A\cap B|\leq 1$ and $|A\cup B|\geq n-1$.
\end{definition}

\begin{example}
Consider $\Gamma(n)$ of $n= 3$.
The set of vertices consists of all eight subsets of $\{1, 2, 3\}$.
By definition of the $\Gamma$-graph, two subsets $A$ and $B$ are adjacent in $\Gamma(3)$ if and only if $\abs{A\cap B} \le 1$ and $\abs{A \cap B} \ge 2$.
We illustrate the graph $\Gamma(3)$ in \Cref{fig:gamma_3}.
As the six vertices $\{1\}, \{2\}, \{3\}, \{1,2\}, \{1,3\}, \{2, 3\}$ form a clique, the chromatic number of $\Gamma(3)$ is at least $6$.
\begin{figure}[h]
\centering
\begin{tikzpicture}[
 agent/.style={regular polygon, regular polygon sides=3, draw, thick, fill=green!30, minimum size=3mm}, 
 item/.style={circle, draw, thick, fill=blue!30, minimum size=5mm}]

\node[] (n1) at ({300}:1.5) {$\{1, 2\}$};
\node[] (n2) at ({0}:1.5) {$\{1, 3\}$};
\node[] (n3) at ({60}:1.5) {$\{2, 3\}$};
\node[] (n4) at ({120}:1.5) {$\{1\}$};
\node[] (n5) at ({180}:1.5) {$\{2\}$};
\node[] (n6) at ({240}:1.5) {$\{3\}$};
\node[] (n0) at (-4, 0) {$\{1,2,3\}$};
\node[] (n7) at (4, 0) {$\emptyset$};

\foreach \i in {1,...,6} {
    \foreach \j in {\i,...,6} {
        \ifnum\i<\j
            \draw[thick] (n\i) -- (n\j);
        \fi
    }
}
\foreach \i in {1,...,3} {
    \draw[thick] (n7) -- (n\i);
}
\foreach \i in {4,...,6} {
    \draw[thick] (n0) -- (n\i);
}
\draw[thick, bend left=-50] (n0) to (n7);
\end{tikzpicture}
\vspace{-8mm}
\caption{Illustration of $\Gamma(3)$}
\label{fig:gamma_3}
\end{figure}
\end{example}

We now elaborate on how $\Gamma$-graph connects to our problem. 
Given the set $M$ of $m$ items, each vertex of $\Gamma(m)$ represents a bundle $A$ and induces an allocation $(A, M\setminus A)$.
By \Cref{prop:ef1_intersection}, for each of $v_1,v_2,u_1$, and $u_2$, if two allocations correspond to two vertices in $\Gamma(m)$ that are adjacent, at least one of them satisfies EF-$1$.
Thus, for each of the four valuations, the set of allocations that fail the EF-1 criterion must form an independent set.
To show the existence of doubly EF-$1$ allocations, it suffices to show that the vertices of $\Gamma(m)$ cannot be covered by four independent sets, i.e., the \emph{chromatic number} of $\Gamma(m)$ is at least $5$. 
Thus, the technical key for proving 
\Cref{thm:double_ef1} is the following theorem.

\begin{theorem}\label{thm:chromatic_gamma_n}
    For every $n\geq 3$, $\chi(\Gamma(n))\geq 5$.
\end{theorem}

We defer the detailed proof to \Cref{sect:proof_chromatic_gamma_n}.
Subsequently, we can prove \Cref{thm:double_ef1} using \Cref{prop:ef1_intersection} and \Cref{sect:proof_chromatic_gamma_n}. 
Assume $m\geq 3$ without loss of generality (for otherwise the theorem is trivial as any allocation $(A,B)$ with $|A|\leq 1$ and $|B|\leq1$ must be doubly EF-$1$).
Consider the graph $\Gamma(m)$ in \Cref{def:gamma_n}, where each vertex corresponds to a subset of items $A$, and it defines an allocation $(A,M\setminus A)$.
Proposition~\ref{prop:ef1_intersection} implies that the allocations that are not EF-$1$ with respect to each of $v_1, v_2, u_1$, and $u_2$ form an independent set in $\Gamma(m)$.
Therefore, as long as $\Gamma(m)$ is not $4$-colorable, there exists an allocation that cannot be covered by the union of the four independent sets, hence doubly EF-$1$.
As shown in Theorem~\ref{thm:chromatic_gamma_n} that $\chi(\Gamma(n))\ge 5$, \Cref{thm:double_ef1} concludes.

\begin{remark}
As a remark, our proof of Theorem~\ref{thm:double_ef1} is non-constructive.
If utility functions are additive, in Sect.~\ref{sec:additive-two-agents}, we show that a doubly EF-$1$ allocation can be computed in polynomial time.
The techniques are different from the ones in this section.    
\end{remark}

\subsection{Proof of Theorem~\ref{thm:chromatic_gamma_n}}
\label{sect:proof_chromatic_gamma_n}
The proof of Theorem~\ref{thm:chromatic_gamma_n} is based on topological arguments and the famous lower bound of \citet{LOVASZ1978319}.
We first review this lower bound and then prove Theorem~\ref{thm:chromatic_gamma_n}.
\Lovaszbound*
By \Cref{thm:neighbor_chrom}, to prove that $\chi(\Gamma(n))\geq 5$, it suffices to show that the geometric realization of the neighborhood complex of $\Gamma(n)$ is $2$-connected.
Denote by $\N$ and $\tilde{\N}$ the neighborhood complex of $\Gamma(n)$ and its geometric realization.
We use $\N_2$ to denote the subcomplex of $\N$ that only contains simplices with dimensions at most $2$ (i.e., $\N_2$ is the $2$-dimensional skeleton of $\N$).
To show that $\tilde{\N}$ is $2$-connected, by \Cref{def:k_connect}, it suffices to show that the injection of $\tilde{\N}_2$ to $\tilde{\N}$ is homotopic to the constant mapping of $\tilde{\N}_2$ to a single point in $\tilde{\N}$ (i.e., the injection of $\tilde{\N}_2$ to $\tilde{\N}$ is a nullhomotopic).

In the following, we will use capital letters $X,Y$ to denote subsets of $[n]$ which are also vertices of both $G$ and $\N$.

\begin{proposition}
\label{prop:form_edge}
For any two vertices $X_1$ and $X_2$, $X_1X_2$ being an edge (i.e., a $1$-dimensional simplex) in $\N$ if and only if $\abs{X_1\setminus X_2}\le 2$ and $\abs{X_2\setminus X_1}\le 2$ simultaneously.
\end{proposition}
\begin{proof}
We first prove that $|X_1\setminus X_2|\leq 2$ and $|X_2\setminus X_1|\leq 2$ if $X_1X_2$ is an edge in $\N$.
Otherwise, if $|X_1\setminus X_2|\geq 3$, then the common neighbor $Y$ of $X_1$ and $X_2$ in the $\Gamma(n)$ graph must contain at most one element in $X_1\setminus X_2$ to ensure $|Y\cap X_1|\leq 1$ and must contain at least two elements in $X_1\setminus X_2$ to ensure $|Y\cup X_2|\geq n-1$, which is impossible.

Second, if $\abs{X_1\setminus X_2}\le 2$ and $\abs{X_2\setminus X_1}\le 2$, take one element $g$ (if it exists) in $X_1\setminus X_2$ and one element $h$ (if it exists) in $X_2\setminus X_1$. 
Let $Y = [n]\setminus (X_1\cup X_2)\cup\{g,h\}$, where the union operation for $g$ or $h$ is omitted if the element does not exist.
Then $Y$ is a common neighbor of $X_1$ and $X_2$ as the union of either $Y$ and $X_1$ or $X_2$ is at least $n-1$, which implies $X_1X_2$ is an edge in $\N$.
\end{proof}

\begin{restatable}[Good edge and Bad edge]{definition}{DefiGood}
\label{def:good_edge_bad_edge}
Given an edge $X_1X_2$ in $\N$, we call it a \emph{good edge} if we have both $|X_1\setminus X_2|\leq 1$ and $|X_2\setminus X_1|\leq 1$.
Otherwise, we call $X_1X_2$ a \emph{bad edge}.
Let $X_1X_2X_3$ be a \emph{good triangle} if all its three edges are good.
Let $X_1X_2X_3$ be a \emph{bad triangle} if one of its three edges is bad.
Let $X_1X_2X_3$ be a \emph{terrible triangle} if at least two of its three edges are bad.    
\end{restatable}
By our observation in \Cref{prop:form_edge} that for every edge $X_1X_2$, $|X_1\setminus X_2|\leq 2$ and $|X_2\setminus X_1|\leq 2$, we have $|X_1\setminus X_2|=2$ or $|X_2\setminus X_1|=2$ for a bad edge $X_1X_2$.
We also have the following proposition.

\begin{proposition}\label{prop:tetrahedron}
If $X_1X_4$, $X_2X_4$, and $X_3X_4$ are good edges of $\N$, then $X_1X_2X_3X_4$ is a tetrahedron in $\N$.
\end{proposition}
\begin{proof}
Let $X_5= [n]\setminus X_4$.
We show that $X_5$ is a common neighbor of all four vertices.
For any of $X_1, X_2, X_3$, as it forms a good edge with $X_4$, it includes at most one less element compared to $X_4$.
Hence, the union of $X_5$ and any of them should be at least $n-1$.
Second, observe that the intersection of $X_5$ and $X_4$ is empty.
In addition, since any of $X_1, X_2, X_3$ includes at most one more element than $X_4$, the intersection of $X_5$ and any of them is at most one.
Therefore, $X_5$ is a common neighbor of all four vertices in the $\Gamma$-graph, which means $X_1X_2X_4X_4$ forms a tetrahedron in $\N$.
\end{proof}

Let $\N_2'$ be the 2-dimensional simplicial complex obtained by removing all the terrible triangles of $\N_2$ (only removing the interior part of the terrible triangles with the skeleton (i.e., edges and vertices) kept).
Let $\N_2''$ be the 2-dimensional simplicial complex obtained by removing all the bad and terrible triangles and all the bad edges of $\N_2$.
The remaining part of the proof consists of three phases.
In the first phase, we show that the injection of $\tilde{\N}_2$ to $\tilde{\N}$ is homotopic to the injection of $\tilde{\N}_2'$ to $\tilde{\N}$.
In the second phase, we show that the injection of $\tilde{\N}_2'$ to $\tilde{\N}$ is homotopic to the injection of $\tilde{\N}_2''$ to $\tilde{\N}$.
In the last phase, we show that the injection of $\tilde{\N}_2''$ to $\tilde{\N}$ is homotopic to the injection of the vertex $\emptyset$ (the vertex representing the empty set of $\{1,\ldots,n\}$) to $\tilde{\N}$.
In other words, we will show in three phases that, within the space $\tilde{\N}$, $\tilde{\N}_2$ can be deformed into $\tilde{\N}_2'$, $\tilde{\N}_2'$ can then be deformed into $\tilde{\N}_2''$, and, lastly, $\tilde{\N}_2''$ can be deformed into a single point. %
In the remaining part of this proof, whenever we say a space $A$ deforms to a space $B$ for $A,B\subseteq\tilde{\N}$, we mean the inclusion map from $A$ to $\tilde{\N}$ is homotopic to the inclusion map from $B$ to $\tilde{\N}$.

\subsubsection{Phase I: $\tilde{\N}_2$ deformed into $\tilde{\N}_2'$}
Let $X_1X_2X_3$ be a terrible triangle and $X_1X_2$ and $X_1X_3$ be its two bad edges.
We find that there always exists a vertex $X_4$ such that $X_1X_4$, $X_2X_4$, and $X_3X_4$ are three edges in $\N$ and they are all good.
The full proof is deferred to \Cref{app:sec3}.
\begin{restatable}{proposition}{propTriTetra}
\label{prop:triangle_tertra}
Let $X_1X_2X_3$ be a terrible triangle with $X_1X_2$ and $X_1X_3$ being its two bad edges.
Then there exists a vertex $X_4$ such that $X_1X_4$, $X_2X_4$, and $X_3X_4$ are three edges in $\N$ and they are all good.
\end{restatable}

Then, by Proposition~\ref{prop:tetrahedron}, $X_1X_2X_3X_4$ is a tetrahedron, and the terrible triangle $X_1X_2X_3$ can be deformed into the union of the remaining three faces $X_1X_2X_4$, $X_1X_3X_4$, and $X_2X_3X_4$ within the tetrahedron $X_1X_2X_3X_4$.
We illustrate the deformation process in Fig.~\ref{fig:deform_n2_n2_prime}.
In addition, to help understand how the bottom triangle face is deformed into the other three faces, we provide a more detailed graphical illustration for the two-dimensional case in Fig.~\ref{fig:deform_triangle_wedge}.
Note that, by our definition, the remaining three faces $X_1X_2X_4$, $X_1X_3X_4$, and $X_2X_3X_4$ are not terrible ($X_1X_2X_4$ and $X_1X_3X_4$ are bad, and $X_2X_3X_4$ can be either good or bad).
Therefore, all the terrible triangles can be removed by deformation within separate tetrahedra.
This completes the description of how $\tilde{\N}_2$ deformed into $\tilde{\N}_2'$ in $\tilde{\N}$.

\begin{figure}[t]
\centering
\begin{tikzpicture}
\pgfmathsetmacro{\shift}{4};

\coordinate (A) at (0,{sqrt(3)});
\coordinate (B) at (-1,0);
\coordinate (C) at (1,0);
\coordinate (Ctrl1) at (-0.5, 0.6); 
\coordinate (Ctrl2) at (0.5, 0.6);

\draw[thick] (A) -- (B);
\draw[thick] (A) -- (C);
\draw[thick] (B) .. controls (Ctrl1) and (Ctrl2) .. (C);
\draw[->, ultra thick,  draw=green!80!black] (0,0.5) -- (0,1);

\node at (2, 1) {\LARGE $\xrightarrow{ \simeq}$};
\node at (-2, 1) {\LARGE $\xrightarrow{\simeq}$};

\pgfmathsetmacro{\shift}{4};
\coordinate (A1) at (0-\shift,{sqrt(3)});
\coordinate (B1) at (-1-\shift,0);
\coordinate (C1) at (1-\shift,0);

\draw[thick] (A1) -- (B1);
\draw[thick] (A1) -- (C1);
\draw[thick] (B1) -- (C1);
\draw[->, ultra thick,  draw=green!80!black] (-\shift,0.1) -- (-\shift,0.6);

\coordinate (A2) at (0+\shift,{sqrt(3)});
\coordinate (B2) at (-1+\shift,0);
\coordinate (C2) at (1+\shift,0);

\draw[very thick] (A2) -- (C2);
\draw[very thick] (A2) -- (B2);
\end{tikzpicture}
\caption{Deformation from a triangle to a wedge}
\label{fig:deform_triangle_wedge}
\end{figure}

\begin{figure}[t]
\centering
\begin{tikzpicture} 
\pgfmathsetmacro{\factor}{1/sqrt(2)};
\pgfmathsetmacro{\shift}{5};
\pgfmathsetmacro{\shifty}{4};

\coordinate [label = right:$X_3$] (A1) at (1-\shift, 0, -1*\factor);
\coordinate [label = left:$X_2$]  (B1) at (-1-\shift, 0, -1*\factor);
\coordinate [label = above:$X_4$] (C1) at (-\shift, 1, 1*\factor) ;
\coordinate [label = below:$X_1$] (D1) at (-\shift, -1, 1*\factor);
\foreach \i in {A1, B1, C1, D1}{
  \draw[dashed] (-\shift, 0) -- (\i);
}

\draw[thick, draw=black!30] (B1) -- (A1);
\path[-, fill = blue!80!black, opacity = .5] (C1) -- (A1) -- (B1) --cycle;
\path[fill = green!50, opacity = .5] (A1) -- (B1) -- (D1) --cycle;
\draw[->, ultra thick, draw=green!80!black] (-\shift + 0.5,-0.2) -- (-\shift + 0.5,0.3);

\draw[thick, draw=blue, dotted] (A1) -- (B1);
\path[-, fill = blue!60, opacity = .5] (B1) -- (D1) -- (C1) --cycle;
\path[-, fill = blue!20, opacity = .5] (C1) -- (D1) -- (A1) --cycle;
\draw[thick, draw=blue] (C1) -- (A1);
\draw[thick, draw=blue] (C1) -- (B1);
\draw[thick, draw=blue] (C1) -- (D1);
\draw[thick, draw=red, decorate, decoration={zigzag}] (B1) -- (D1);
\draw[thick, draw=red, decorate, decoration={zigzag}] (D1) -- (A1);

\node at (-\shift, -2.5) {$\tilde{N}_2$};

\node at (-\shift/2, 0) {\LARGE $\xrightarrow{\simeq}$};

\coordinate [label = right:$X_3$] (A2) at (1, 0, -1*\factor);
\coordinate [label = left:$X_2$]  (B2) at (-1, 0, -1*\factor);
\coordinate [label = above:$X_4$] (C2) at (0, 1, 1*\factor) ;
\coordinate [label = below:$X_1$] (D2) at (0, -1, 1*\factor);
\foreach \i in {A2, B2, C2, D2}{
  \draw[dashed] (0, 0) -- (\i);
}

\draw[thick, dotted] (B2) -- (A2);
\path[pattern=north east lines, pattern color=black, opacity=0.6] (A2) -- (B2) -- (D2) --cycle;
\path[-, fill = blue!80!black, opacity = .5] (C2) -- (A2) -- (B2) --cycle;
\path[-, fill = blue!60, opacity = .5] (B2) -- (D2) -- (C2) --cycle;
\path[-, fill = blue!20, opacity = .5] (C2) -- (D2) -- (A2) --cycle;
\draw[thick, draw=blue] (C2) -- (A2);
\draw[thick, draw=blue] (C2) -- (B2);
\draw[thick, draw=blue] (C2) -- (D2);
\draw[thick, draw=red, decorate, decoration={zigzag}] (B2) -- (D2);
\draw[thick, draw=red, decorate, decoration={zigzag}] (D2) -- (A2);
\node at (0, -2.5) {$\tilde{N}_2'$};
\end{tikzpicture}
\caption{The deformation from $\tilde{\N}_2$ to $\tilde{\N}_2'$, where the blue edges $X_1X_4, X_2X_4$, and $X_3X_4$ represent good edges, the zigzag edges $X_1X_2, X_2X_3$ represent bad edges, and the black edge $X_1X_3$ can be either good or bad.
The interior part of the terrible triangle $X_1X_2X_3$ is removed in the right 2-dimensional complex.
}
\label{fig:deform_n2_n2_prime}
\end{figure}

\subsubsection{Phase II: $\tilde{\N}_2'$ deformed into $\tilde{\N}_2''$}
We next consider a deformation from $\tilde{\N}_2'$ (where all terrible triangles are removed) to $\tilde{\N}_2''$ by discussing the three types of bad edges $X_1X_2$:
\begin{itemize}
    \item Type 1: $|X_1\setminus X_2|=|X_2\setminus X_1|=2$.
    \item Type 2: $|X_1\setminus X_2|=1$ and $|X_2\setminus X_1|=2$.
    \item Type 3: $|X_1\setminus X_2|=0$ and $|X_2\setminus X_1|=2$.
\end{itemize}

First, consider type 1 edges.
Given a type 1 edge $X_1X_2$, let $\{a,b\}=X_1\setminus X_2$ and $\{c,d\}=X_2\setminus X_1$.
Let $X=X_1\cap X_2$.
We can observe that there are exactly four bad triangles incident on $X_1X_2$.
In particular, any bad triangle $X_1X_2Y$ in $\tilde{N}_2'$ cannot be terrible by the definition of $\tilde{N}_2'$.
Therefore, both $X_1Y$ and $X_2Y$ are good edges, which implies that $\abs{X_1\setminus Y}, \abs{X_2\setminus Y}, \abs{Y\setminus X_1}, \abs{Y\setminus X_2} \le 1$.
Thus, $Y$ must contain exactly one of $a,b$, exactly one of $c,d$, all of $X$, and none of $[n]\setminus(X\cup\{a,b,c,d\})$.
Then the four bad triangles incident to $X_1X_2$ are as follows (also as shown in Fig.~\ref{fig:type1_deform}):
\begin{alignat*}{3}
& X_1 X_2 X_3 \quad \text{with $X_3 = X \cup \{a, c\}$} \quad 
&& X_1 X_2 X_4 \quad \text{with $X_4 = X \cup \{a, d\}$} \\
& X_1X_2X_5 \quad \text{with $X_5 = X \cup \{b,c\}$}
&& X_1X_2X_6 \quad \text{with $X_6 = X \cup \{b,d\}$}
\end{alignat*}
It is also easy to check that $X_1X_2X_3X_4$, $X_1X_2X_3X_5$, $X_1X_2X_4X_6$, and $X_1X_2X_5X_6$ are tetrahedra of $\N$.
In addition, all the edges except for $X_1X_2$ in each of the tetrahedra are good.

As illustrated in Fig.~\ref{fig:type1_deform}, we first construct a deformation of the interior of $X_1X_2X_3$ into the other three faces of the tetrahedron $X_1X_2X_3X_4$, just as we did in the first phase, which removes the interior part of $X_1X_2X_3$.
Then, we remove the triangle $X_1X_2X_4$ by considering the same kind of deformation in the tetrahedron $X_1X_2X_4X_6$.
Next, we remove the triangle $X_1X_2X_6$ by considering the deformation in the tetrahedron $X_1X_2X_5X_6$.
Now, only one bad triangle incident to $X_1X_2$ remains, namely, $X_1X_2X_5$.
Finally, we consider a deformation of $X_1X_2$ into the other two edges $X_1X_5$ and $X_2X_5$ within the triangle $X_1X_2X_5$.
This removes the bad edge $X_1X_2$ and all the bad triangles incident to it.
Notice that $X_3X_6$ and $X_4X_5$ are bad edges.
Although $X_1X_2X_3X_6$ and $X_1X_2X_4X_5$ are two tetrahedra of $\N$ (e.g., $([n]\setminus X)\cup\{a,d\}$ is a common neighbor of $X_1,X_2,X_3$, and $X_6$), we have avoided using them in the operations above.
The reason we need to avoid this is the following:
the bad edge $X_3X_6$ may have been removed prior to $X_1X_2$;
If we deform the interior of $X_1X_2X_3$ to the other three faces of $X_1X_2X_3X_6$, the other three faces will reappear, which makes $X_3X_6$ reappear.

\begin{figure}[h]
\centering
\begin{tikzpicture}
\coordinate (C1) at (0, 0);
\coordinate (C2) at (0, -2);
\coordinate (C3) at (1, 0.5);
\coordinate (C4) at (2, -0.5);
\coordinate (C5) at (2, -1.5);
\coordinate (C6) at (1, -2.5);

\node[circle, fill=black,inner sep=1pt,label={[left]{$X_1$}}] (X1) at (C1) {};
\node[circle, fill=black,inner sep=1pt,label={[left]{$X_2$}}] (X2) at (C2) {};
\node[circle, fill=black,inner sep=1pt,label={[right]{$X_3$}}] (X3) at (C3) {};

\node[circle, fill=black,inner sep=1pt,label={[right]{$X_4$}}] (X4) at (C4) {};
\node[circle, fill=black,inner sep=1pt,label={[right]{$X_5$}}] (X5) at (C5) {};
\node[circle, fill=black,inner sep=1pt,label={[right]{$X_6$}}] (X6) at (C6) {};

\draw[draw=red, decorate, decoration={zigzag}] (X4) -- (X5);
\draw[draw=red, decorate, decoration={zigzag}] (X3) -- (X6);

\draw[thick, draw=blue, dotted] (X2) -- (X3);

\path[fill=green!80!black, opacity=0.4] (C1) -- (C2) -- (C3) -- cycle;
\draw[->, ultra thick,  draw=green!60!black] (0.3,-0.3) -- (0.8,-0.4);

\path[fill=blue!60, opacity=0.4] (C2) -- (C3) -- (C4) -- cycle;
\path[fill=blue!80, opacity=0.4] (C1) -- (C3) -- (C4) -- cycle;
\path[fill=blue!20, opacity=0.4] (C1) -- (C2) -- (C4) -- cycle;

\draw[draw=blue] (X4) -- (X6);
\draw[draw=blue] (X1) -- (X5);
\draw[draw=blue] (X2) -- (X5);
\draw[draw=blue] (X3) -- (X5);
\draw[draw=blue] (X1) -- (X6);
\draw[draw=blue] (X2) -- (X6);
\draw[thick, draw=blue] (X3) -- (X4);
\draw[thick, draw=blue] (X1) -- (X3);
\draw[thick, draw=blue] (X1) -- (X4);
\draw[thick, draw=blue] (X2) -- (X4);
\draw[draw=blue] (X5) -- (X6);
\draw[thick, draw=red, decorate, decoration={zigzag}] (X1) -- (X2);

\node[] at (3.5, -0.75) {\LARGE $\xrightarrow{\simeq}$};

\pgfmathsetmacro{\shift}{5};
\coordinate (D1) at (\shift, 0);
\coordinate (D2) at (\shift, -2);
\coordinate (D3) at (1+\shift, 0.5);
\coordinate (D4) at (2+\shift, -0.5);
\coordinate (D5) at (2+\shift, -1.5);
\coordinate (D6) at (1+\shift, -2.5);

\node[circle, fill=black,inner sep=1pt,label={[left]{$X_1$}}] (Y1) at (D1) {};
\node[circle, fill=black,inner sep=1pt,label={[left]{$X_2$}}] (Y2) at (D2) {};
\node[circle, fill=black,inner sep=1pt,label={[right]{$X_3$}}] (Y3) at (D3) {};
\node[circle, fill=black,inner sep=1pt,label={[right]{$X_4$}}] (Y4) at (D4) {};
\node[circle, fill=black,inner sep=1pt,label={[right]{$X_5$}}] (Y5) at (D5) {};
\node[circle, fill=black,inner sep=1pt,label={[right]{$X_6$}}] (Y6) at (D6) {};

\draw[draw=red, decorate, decoration={zigzag}] (Y3) -- (Y6);
\draw[draw=blue] (Y2) -- (Y3);
\path[pattern=north east lines, pattern color=black, opacity=0.6] (D1) -- (D2) -- (D3) --cycle;

\draw[draw=blue] (Y3) -- (Y5);
\draw[draw=blue] (Y1) -- (Y3);
\draw[thick, draw=blue, dotted] (Y2) -- (Y4);
\draw[draw=blue] (Y1) -- (Y5);
\draw[draw=blue] (Y2) -- (Y5);
\draw[draw=blue] (Y3) -- (Y4);
\draw[draw=red, decorate, decoration={zigzag}] (Y4) -- (Y5);
\draw[draw=blue] (Y5) -- (Y6);

\path[fill=green!80!black, opacity=0.4] (D1) -- (D4) -- (D2) -- cycle;
\draw[->, ultra thick,  draw=green!60!black] (0.6+\shift,-1) -- (0.9+\shift,-1.4);
\path[fill=blue!60, opacity=0.4] (D2) -- (D4) -- (D6) -- cycle;
\path[fill=blue!20, opacity=0.4] (D1) -- (D4) -- (D6) -- cycle;
\path[fill=blue!80, opacity=0.4] (D1) -- (D2) -- (D6) -- cycle;

\draw[thick, draw=blue] (Y1) -- (Y6);
\draw[thick, draw=blue] (Y4) -- (Y6);
\draw[thick, draw=blue] (Y2) -- (Y6);
\draw[thick, draw=blue] (Y1) -- (Y4);
\draw[thick, draw=red, decorate, decoration={zigzag}] (Y1) -- (Y2);

\node[] at (3.5+\shift, -0.75) {\LARGE $\xrightarrow{\simeq}$};

\pgfmathsetmacro{\shiftp}{2*\shift};
\coordinate (E1) at (0+\shiftp, 0);
\coordinate (E2) at (0+\shiftp, -2);
\coordinate (E3) at (1+\shiftp, 0.5);
\coordinate (E4) at (2+\shiftp, -0.5);
\coordinate (E5) at (2+\shiftp, -1.5);
\coordinate (E6) at (1+\shiftp, -2.5);

\node[circle, fill=black,inner sep=1pt,label={[left]{$X_1$}}] (W1) at (E1) {};
\node[circle, fill=black,inner sep=1pt,label={[left]{$X_2$}}] (W2) at (E2) {};
\node[circle, fill=black,inner sep=1pt,label={[right]{$X_3$}}] (W3) at (E3) {};
\node[circle, fill=black,inner sep=1pt,label={[right]{$X_4$}}] (W4) at (E4) {};
\node[circle, fill=black,inner sep=1pt,label={[right]{$X_5$}}] (W5) at (E5) {};
\node[circle, fill=black,inner sep=1pt,label={[right]{$X_6$}}] (W6) at (E6) {};

\draw[draw=blue] (W3) -- (W5);
\draw[draw=blue] (W4) -- (W6);
\draw[draw=blue] (W2) -- (W3);
\draw[draw=red, decorate, decoration={zigzag}] (W3) -- (W6);
\draw[draw=red, decorate, decoration={zigzag}] (W4) -- (W5);

\path[pattern=north east lines, pattern color=black, opacity=0.6] (E1) -- (E2) -- (E3) --cycle;
\path[pattern=north east lines, pattern color=black, opacity=0.6] (E1) -- (E2) -- (E4) --cycle;

\draw[draw=blue] (W1) -- (W3);
\draw[draw=blue] (W1) -- (W4);
\draw[draw=blue] (W2) -- (W4);
\draw[thick, draw=blue, dotted] (W2) -- (W5);
\draw[draw=blue] (W3) -- (W4);

\path[fill=blue!80!black, opacity=0.4] (E1) -- (E2) -- (E5) -- cycle;
\path[fill=blue!60, opacity=0.4] (E2) -- (E5) -- (E6) -- cycle;
\draw[->, ultra thick,  draw=green!60!black] (0.4+\shiftp,-2) -- (0.9+\shiftp,-1.8);
\path[fill=green!80, opacity=0.4] (E1) -- (E2) -- (E6) -- cycle;
\path[fill=blue!5, opacity=0.4] (E1) -- (E5) -- (E6) -- cycle;

\draw[thick,draw=blue] (W1) -- (W6);
\draw[thick,draw=blue] (W2) -- (W6);
\draw[thick,draw=blue] (W5) -- (W6);
\draw[thick, draw=blue] (W1) -- (W5);
\draw[thick, draw=red, decorate, decoration={zigzag}] (W1) -- (W2);

\pgfmathsetmacro{\shift}{2};
\pgfmathsetmacro{\shiftx}{4};
\pgfmathsetmacro{\shifty}{4};

\node[] at (\shift+\shiftx, -0.75-\shifty) {\LARGE $\xrightarrow{\simeq}$};

\coordinate (H1) at (0+\shift, -\shifty);
\coordinate (H2) at (0+\shift, -2-\shifty);
\coordinate (H3) at (1+\shift, 0.5-\shifty);
\coordinate (H4) at (2+\shift, -0.5-\shifty);
\coordinate (H5) at (2+\shift, -1.5-\shifty);
\coordinate (H6) at (1+\shift, -2.5-\shifty);

\draw[->, ultra thick,  draw=green!60!black] (0.1+\shift,-1-\shifty) -- (0.6+\shift,-1-\shifty);

\node[circle, fill=black,inner sep=1pt,label={[left]{$X_1$}}] (K1) at (H1) {};
\node[circle, fill=black,inner sep=1pt,label={[left]{$X_2$}}] (K2) at (H2) {};
\node[circle, fill=black,inner sep=1pt,label={[right]{$X_3$}}] (K3) at (H3) {};
\node[circle, fill=black,inner sep=1pt,label={[right]{$X_4$}}] (K4) at (H4) {};
\node[circle, fill=black,inner sep=1pt,label={[right]{$X_5$}}] (K5) at (H5) {};
\node[circle, fill=black,inner sep=1pt,label={[right]{$X_6$}}] (K6) at (H6) {};

\draw[draw=blue] (K3) -- (K5);
\draw[draw=blue] (K4) -- (K6);
\draw[draw=blue] (K2) -- (K3);
\draw[draw=red, decorate, decoration={zigzag}] (K3) -- (K6);
\draw[draw=red, decorate, decoration={zigzag}] (K4) -- (K5);

\path[pattern=north east lines, pattern color=black, opacity=0.6] (H1) -- (H2) -- (H3) --cycle;
\path[pattern=north east lines, pattern color=black, opacity=0.6] (H1) -- (H2) -- (H4) --cycle;
\path[pattern=north east lines, pattern color=black, opacity=0.6] (H1) -- (H2) -- (H6) --cycle;
\path[fill=green!20, opacity=0.4] (H1) -- (H2) -- (H5) -- cycle;

\draw[draw=blue] (K1) -- (K3);
\draw[draw=blue] (K1) -- (K4);
\draw[draw=blue] (K2) -- (K4);
\draw[thick, draw=blue] (K1) -- (K5);
\draw[thick, draw=blue] (K2) -- (K5);
\draw[draw=blue] (K1) -- (K6);
\draw[draw=blue] (K2) -- (K6);
\draw[draw=blue] (K3) -- (K4);
\draw[draw=blue] (K5) -- (K6);
\draw[thick, draw=red, decorate, decoration={zigzag}] (K1) -- (K2);

\pgfmathsetmacro{\shift}{4};
\coordinate (F1) at (\shift*2, 0-\shifty);
\coordinate (F2) at (\shift*2, -2-\shifty);
\coordinate (F3) at (1+\shift*2, 0.5-\shifty);
\coordinate (F4) at (2+\shift*2, -0.5-\shifty);
\coordinate (F5) at (2+\shift*2, -1.5-\shifty);
\coordinate (F6) at (1+\shift*2, -2.5-\shifty);

\node[circle, fill=black,inner sep=1pt,label={[left]{$X_1$}}] (W1) at (F1) {};
\node[circle, fill=black,inner sep=1pt,label={[left]{$X_2$}}] (W2) at (F2) {};
\node[circle, fill=black,inner sep=1pt,label={[right]{$X_3$}}] (W3) at (F3) {};
\node[circle, fill=black,inner sep=1pt,label={[right]{$X_4$}}] (W4) at (F4) {};
\node[circle, fill=black,inner sep=1pt,label={[right]{$X_5$}}] (W5) at (F5) {};
\node[circle, fill=black,inner sep=1pt,label={[right]{$X_6$}}] (W6) at (F6) {};

\draw[draw=blue] (W2) -- (W3);
\draw[draw=blue] (W3) -- (W5);
\draw[draw=blue] (W4) -- (W6);

\path[pattern=north east lines, pattern color=black, opacity=0.6] (F1) -- (F2) -- (F3) --cycle;
\path[pattern=north east lines, pattern color=black, opacity=0.6] (F1) -- (F2) -- (F4) --cycle;
\path[pattern=north east lines, pattern color=black, opacity=0.6] (F1) -- (F2) -- (F6) --cycle;
\path[pattern=north east lines, pattern color=black, opacity=0.6] (F1) -- (F2) -- (F5) --cycle;

\draw[draw=red, decorate, decoration={zigzag}] (W3) -- (W6);
\draw[draw=blue] (W1) -- (W3);
\draw[draw=blue] (W1) -- (W4);
\draw[draw=blue] (W2) -- (W4);
\draw[draw=blue] (W1) -- (W5);
\draw[draw=blue] (W2) -- (W5);
\draw[draw=blue] (W1) -- (W6);
\draw[draw=blue] (W2) -- (W6);
\draw[draw=blue] (W3) -- (W4);
\draw[draw=red, decorate, decoration={zigzag}] (W4) -- (W5);
\draw[draw=blue] (W5) -- (W6);

\end{tikzpicture}
\caption{Deformation from $\tilde{\N}_2'$ to $\tilde{\N}_2''$ for type one bad edge $X_1X_2$, where gray hatched areas represent the removed faces, red zigzag lines represent bad edges, and both dotted and solid lines represent good edges
}
\label{fig:type1_deform}
\end{figure}

Next, we consider type 2 edges.
Let $X_1X_2$ be a type 2 edge. Let $\{a\}=X_1\setminus X_2$ and $\{b,c\}=X_2\setminus X_1$.
Let $X=X_1\cap X_2$.
It is easy to check that there are exactly four bad triangles incident to $X_1X_2$:
\begin{alignat*}{3}
& X_1 X_2 X_3 \quad \text{with $X_3 = X \cup \{a, b\}$} \quad 
&& X_1 X_2 X_4 \quad \text{with $X_4 = X \cup \{a, c\}$} \\
& X_1X_2X_5 \quad \text{with $X_5 = X \cup \{b\}$}
&& X_1X_2X_6 \quad \text{with $X_6 = X \cup \{c\}$}
\end{alignat*}
It is also easy to check that $X_1X_2X_3X_4$, $X_1X_2X_3X_5$, $X_1X_2X_4X_6$, and $X_1X_2X_5X_6$ are tetrahedra of $\N$.
In addition, all the edges except for $X_1X_2$ in each of the tetrahedra are good.
The remaining procedure is the same as it is in handling Type 1 edges.
Eventually, we can remove $X_1X_2$ and all the bad triangles incident to it.

Finally, consider a type 3 edge $X_1X_2$ and let $\{a,b\}=X_2\setminus X_1$.
Similarly, there are only two bad triangles incident to $X_1X_2$: 
\begin{align*}
X_1X_2X_3\quad \text{with $X_3=X_1\cup\{a\}$}\qquad X_1X_2X_4 \quad \text{with $X_4=X_1\cup\{b\}$}
\end{align*}
It can be verified that $X_1X_2X_3X_4$ is a tetrahedron and $X_1X_2$ is the only bad edge in it.
We can first remove the interior of $X_1X_2X_3$ by letting it deform into the other three faces of $X_1X_2X_3X_4$, and then remove the edge $X_1X_2$ and the remaining bad triangle $X_1X_2X_4$ by considering the deformation of $X_1X_2$ into the union of $X_1X_4$ and $X_2X_4$.

For each bad edge $X_1X_2$ (of each type), the deformation described above only happens in tetrahedra where $X_1X_2$ is the unique bad edge.
Therefore, all these deformations happen in separate regions.
We have described how $\tilde{\N}_2'$ is deformed into $\tilde{\N}_2''$ in $\tilde{\N}$.

\subsubsection{Phase III: $\tilde{\N}_2''$ deformed into the point $\emptyset$}
We perform the final deformation by $n$ steps.
We first consider removing all vertices $X$ with $n\in X$, as well as the edges and triangles incident to them.
Consider the mapping $\phi_n$ that maps a vertex to another defined by the following:
$$\phi_n(X)=X\setminus\{n\}.$$
We show that $\phi_n$ can be extended to a deformation in $\tilde{\N}$ from $\tilde{\N}_2''$ to a subcomplex where all the aforementioned simplices related to $n$ are removed.

\begin{itemize}[leftmargin=*]
    \item (1-dimensional case): For each edge $X_1X_2$ in $\N_2''$, we will show that 1)  $X_1,X_2,\phi_n(X_1),\phi_n(X_2)$ span a simplex (i.e., a tetrahedron) in $\N$, so that $X_1X_2$ can be continuously ``moved'' to $\phi_n(X_1)\phi_n(X_2)$ within the simplex, and 2) $\phi_n(X_1)\phi_n(X_2)$ is a good edge in $\N_2''$.
    \item (2-dimensional case): We also need to show that the same holds for 2-dimensional simplices.
    For each triangle $X_1X_2X_3$ in $\N_2''$, we will show that 1) $X_1,X_2,X_3,\phi_n(X_1),\phi_n(X_2),\phi_n(X_3)$ span a simplex in $\N$, so that $X_1X_2X_3$ can be continuously ``moved'' to $\phi_n(X_1),\phi_n(X_2),\phi_n(X_3)$ within the simplex, and 2) $\phi_n(X_1)\phi_n(X_2)\phi_n(X_3)$ is a good triangle in $\N_2''$. 
\end{itemize}

Note also that all edges and triangles in $\N_2''$ are good.
It is straightforward to check that point 2) for both 1-dimensional case and 2-dimensional case always holds (this is because $|(X_1\setminus \{g\})\setminus (X_2\setminus \{g\})|\leq |X_1\setminus X_2|$ holds for any two sets $X_1,X_2$ and any element $g$, and all the good edges of $\N_2$ remains in $\N_2''$ as we have not removed any good edge in the first two phases).
It remains to check point 1).

We first check for the 1-dimensional case.
If $X_1X_2$ is a good edge, then $|X_1\setminus X_2|\leq 1$ and $|X_2\setminus X_1|\leq 1$, and $Y=[n]\setminus(X_1\cap X_2)$ is a common neighbor for $X_1$ and $X_2$, with $|Y\cup X_1|=|Y\cup X_2|=n$.
Since $\phi_n(X_1)\subseteq X_1$ and $|X_1|-1\leq |\phi_n(X_1)|\leq |X_1|$, we have $|\phi_n(X_1)\cap Y|\leq |X_1\cap Y|\leq 1$ and $|\phi_n(X_1)\cup Y|\geq n-1$, so $\phi_n(X_1)$ is a neighbor of $Y$.
Similarly, $\phi_n(X_2)$ is also a neighbor of $Y$.
Therefore, $X_1,X_2,\phi_n(X_1),\phi_n(X_2)$ span a simplex in $\N$.

We next check for the 2-dimensional case.
Let $X_1X_2X_3$ be a good triangle.
Let $X=X_1\cap X_2\cap X_3$.
First of all, we must have $|X_i\setminus X|\leq 2$ for each $i\in\{1,2,3\}$.
Suppose otherwise $|X_1\setminus X|\geq 3$. Since $X_2X_1$ and $X_3X_1$ are good edges, each of $X_2$ and $X_3$ must contain at least $|X_1\setminus X|-1$ elements in $X_1\setminus X$, and one of the elements in $X_1\setminus X$ must be contained in all of $X_1,X_2,X_3$, which contradicts to our definition of $X$.
We then discuss two cases.
\begin{itemize}[leftmargin=*]
    \item Case 1: $|X_i\setminus X|\leq 1$ for each $i\in\{1,2,3\}$. In this case, $Y=[n]\setminus X$ is a common neighbor for all three vertices since $|Y\cup X_i|=n$ and $\abs{Y\cap X_i}\le 1$ for each $i$. 
    For similar reasons as in the 1-dimensional case, we can show that $Y$ is a neighbor to $\phi_n(X_i)$ for each $i$. Therefore, $X_1,X_2,X_3,\phi_n(X_1),\phi_n(X_2),\phi_n(X_3)$ span a simplex in $\N$.
    \item Case 2: $|X_1\setminus X|=2$. Let $\{a,b\}=X_1\setminus X$. Since $X_2X_1$ and $X_3X_1$ are good edges, it must be that $X_2$ contains exactly one of $a,b$, and $X_3$ contains exactly the other one. (If $X_2$ contains none of $a,b$, then $X_2X_1$ is a bad edge; if $X_2$ contains both of $a,b$, then $X_3$ must contain at least one of them, which implies one of $a,b$ is in all the three sets, which contradicts to our definition of $X$.)
    Assume that $a\in X_2$ and $b\in X_3$.
    Notice that, for now, we already have $|X_2\setminus X_3|=|X_3\setminus X_2|=|X_1\setminus X_2|=|X_1\setminus X_3|=1$.
    There are only two possibilities:
    \begin{enumerate}[leftmargin=*]
        \item $X_1=X\cup\{a,b\}, X_2=X\cup\{a\}, X_3=X\cup\{b\}$;
        \item there exists a third element $c$ such that $X_1=X\cup\{a,b\}, X_2=X\cup\{a,c\}, X_3=X\cup\{b,c\}$.
    \end{enumerate}
    For the first possible case, consider $Y=[n]\setminus(X\cup\{a,b\})\cup\{n\}$.
    For each of $X_i$, we have $Y\cap X_i \le 1$ (the intersection at most contains element $n$) and $\abs{Y\cup X_i} \ge n-1$.
    Similarly, for each of $\phi(X_i)$, we  have $Y\cap X_i = \emptyset$ and $\abs{Y\cup X_i} \ge n-1$.
    Therefore, we can see that $Y$ is a common neighbor for $X_1,X_2,X_3,\phi_n(X_1),\phi_n(X_2),\phi_n(X_3)$; for the second possibility, similarly, $Y=[n]\setminus(X\cup\{a,b,c\})\cup\{n\}$ is a common neighbor for $X_1,X_2,X_3,\phi_n(X_1),\phi_n(X_2),\phi_n(X_3)$.
    Notice that element $n$ may be one of $a,b,c$, but it does not invalidate the above argument.
\end{itemize}

We have shown that there exists a mapping $\phi_n$ demonstrating the deformation of $\tilde{\N}_2''$ into a subcomplex (of $\tilde{\N}_2''$) with only $2^{n-1}$ vertices corresponding to subsets of $[n-1]$.
By similarly defining $\phi_{n-1},\phi_{n-2},\ldots,\phi_1$ and iteratively applying them, $\tilde{\N}_2''$ can be deformed into $\emptyset$ in the space $\tilde{\N}$.

\subsection{Correction of WINE'23 Conference Version}
In the WINE'23 conference version~\citep{bu2023fair}, we claim the same result as in \Cref{thm:double_ef1} using the chromatic number of generalized Kneser graphs (see Definition~\ref{def:generalize_kneser_graph}).
However, our proof in the WINE'23 conference version uses an unproven conjecture. We did not realize it was a conjecture, and we mistakenly thought it had been proved before.
We elaborate on these below.

It is easy to see that the generalized Kneser graph $\K(2k,k,1)$ is a subgraph of $\Gamma(2k)$.
In the language of fair divisions, vertices of $\Gamma(2k)$ are all possible bundles of $2k$ items, whereas vertices of $\K(2k,k,1)$ are bundles with exactly $k$ items.
As a result, vertices of $\K(2k,k,1)$ are all \emph{balanced} allocations $(A,B)$ with $|A|=|B|=k$.
Notice that the definition of edges in both $\K(2k,k,1)$ and $\Gamma(2k)$ is the same.
Therefore, for each of $v_1,v_2,u_1$, and $u_2$, the set of balanced allocations that fail the EF-$1$ criterion forms an independent set of $\K(2k,k,1)$.
By the same analysis before, we can show the existence of a balanced doubly EF-$1$ allocation if the chromatic number of $\K(2k,k,1)$ is at least $5$.
However, the best lower bound known for the chromatic number of the generalized Kneser graph is given in Lemma~\ref{lem:chromatic_for_kneser}. This bound only implies $\chi(2k,k,1)\geq 4$, which is not sufficient for our application. 

\citet{jafari2017chromatic} conjectured the following.

\begin{conjecture}[\citet{jafari2017chromatic}]\label{conj:kneser_conjecture}
    For every $k\ge 3$, $\chi(2k,k,1)=6$.
\end{conjecture}

If the conjecture holds, Theorem~\ref{thm:double_ef1} follows.
In our WINE'23 version, we mistakenly present the above conjecture as a proved theorem.

\section{Double Fairness with Additive Utilities}
In this section, we present the results of double fairness. 
Utility functions are assumed to be additive throughout this section.

\subsection{Identical Allocator's Utility Function}\label{sect:identical_double_ef1}
This section considers the case when the allocator's utilities $u_1,\ldots,u_n$ are identical.
Let $u=u_1=\cdots=u_n$.

We first give a brief introduction on the techniques used in this section. 
The envy-cycle procedure was first proposed by~\citet{Lipton04onapproximately} to compute an EF-$1$ allocation for general valuations.
In the envy-cycle procedure, an \emph{envy-graph} is constructed for a partial allocation.
Each vertex in the envy-graph represents an agent and each directed edge $(u, v)$ means that agent $u$ envies agent $v$ in the current allocation.
An item is received by a source agent (an agent without incoming-edge) in each round.
When there is a cycle in an envy-graph, we use the cycle-elimination algorithm to eliminate this cycle while maintaining the EF-$1$ property (as the allocation after cycle-elimination is a permutation of the previous allocation, and each agent's utility strictly increases).
\begin{definition}[Cycle-Elimination Algorithm]
Given an envy-graph with a cycle $u_1\rightarrow \ldots \rightarrow u_n \rightarrow u_1$, shift the agents' bundles along the cycle ($A_{u_i} \leftarrow A_{u_{i+1}}$ for $i=1, \ldots, n-1$ and $A_n \leftarrow A_1$). 
\end{definition}

\begin{algorithm}[t]
\caption{Finding doubly EF-$1$ allocation for identical utility functions}\label{alg:identical_2ef1}
Let $G=(V, E)$ be the envy-graph where each vertex represents an agent and $E\leftarrow \emptyset$\;
Initialize $\mathcal{A} = \left(\emptyset, \ldots, \emptyset \right)$\;
\If{$m$ is not divided by $n$}{
Add dummy items to $M$ such that $n\mid m$ and set the utility of each dummy item as $0$\;
}
Let $M_s$ be the sorted array of the items according to the allocator's utility function $u$ in descending order\;

\For{every $n$ items $M_n \subseteq M_s$}{
    Let $\{i_1, \dots, i_n\}$ be the agents in topological order of graph $G$\;
    \For{each $j\in \{1, \dots, n\}$}{
        Allocate agent $i_j$'s favorite item $g\in M_n$ to $i_j$: $A_{i_j} \leftarrow A_{i_j}\cup\{\mathop{\mathrm{argmax}}_{g\in M_n}{v_{i_j}(g)}\}$\;
        $M_n\leftarrow M_n\setminus \{g\}$\;
    }
    Update the envy-graph $G$\;
    Iteratively run the cycle-elimination algorithm and update $G$ until $G$ contains no cycle\;
}
Remove the dummy items from the allocation $\mathcal{A}$ \;
\Return{the allocation $\mathcal{A}$}
\end{algorithm}

\begin{theorem} \label{thm:identical_2ef1}
When the allocator's utility functions are identical, a doubly EF-$1$ allocation always exists for any number of agents $n$, and can be found by Algorithm~\ref{alg:identical_2ef1} in polynomial time. 
\end{theorem}

Before we prove Theorem~\ref{thm:identical_2ef1}, we first describe our algorithm. 
At the beginning of the algorithm, we construct an envy-graph $G$ with $n$ vertices and no edges and sort the items according to the allocator's utility function in descending order.
Then, we divide the sorted items into $\ceil{\frac{m}{n}}$ groups where each group contains $n$ items.
In each round, we allocate a group of items to the agents such that each agent receives exactly one item.
In particular, each agent takes away her favorite item from the group, where the agents are sorted in the topological order of $G$ before the iteration begins.
After all these $n$ items are allocated, we update the envy-graph and run the cycle-elimination algorithm, so that the envy-graph contains no cycle and a topological order of the agents can be successfully found in the next round.

To prove Theorem~\ref{thm:identical_2ef1}, we first prove the allocation is EF-$1$ from both the agents' and the allocator's perspectives.

\begin{lemma}
    The allocation computed by Algorithm~\ref{alg:identical_2ef1} is EF-$1$ to the agents.
\end{lemma}
\begin{proof}
We use induction to show that EF-$1$ is maintained to the agents through the algorithm. 
At the beginning of the first round, the allocation is empty, so it is EF-$1$.
We assume at the beginning of the $\ell$-th round, the allocation $\mathcal{A}$ is EF-$1$.
Denote the allocation after running the $(\ell+1)$-th round by $\mathcal{B}$.
We now show $\mathcal{B}$ is still EF-$1$.

We first consider the allocation $\mathcal{A'}$ before the cycle-elimination algorithm. 
Consider two arbitrary agents $i, j$, and assume they receive items $g_i$ and $g_j$ respectively in the $(\ell+1)$-th round. 
Without loss of generality, we assume $i$ is before $j$ in the topological order of $G$ after the $\ell$-th round.
For agent $i$, since $\mathcal{A}$ is EF-$1$, there exists an item $g\in A_j$ such that $v_i(A_i)\ge v_i\left(A_j\setminus \{g\}\right)$.
Since $i$ is before $j$, $v_i(g_i)\ge v_i(g_j)$.
We have 
$$
v_i(A'_i)=v_i\left(A_i\cup \{g_i\}\right) \ge v_i((A_j\setminus \{g\})\cup\{g_j\})=v_i(A'_j\setminus\{g\}),
$$
so agent $i$ will not envy agent $j$ if $g$ is removed from $A'_j$.
For agent $j$, she does not envy $i$ in $\mathcal{A}$, so $v_j\left(A_j\right)\ge v_j(A_i)$. 
Then, we have 
$$
v_j(A'_j)=v_j\left(A_j\cup\{g_j\}\right)\ge v_j(A_i)=v_j(A'_i\setminus\{g_i\}),
$$
so $j$ will not envy $i$ if $g_i$ is removed from $A'_i$.
Hence, $\mathcal{A'}$ is an EF-$1$ allocation to the agents.

The cycle-elimination algorithm does not destroy the EF-$1$ property. The allocation $\mathcal{B}$ after cycle-elimination is a permutation of $\mathcal{A'}$ where the constituents of each bundle do not change, and each agent receives a bundle with a weakly higher value. Hence, $\mathcal{B}$ is still EF-$1$ to the agents.
\end{proof}

\begin{lemma} \label{lem:identical_allocator_ef1}
    The allocation found by Algorithm~\ref{alg:identical_2ef1} is EF-$1$ to the allocator.
\end{lemma}
\begin{proof}
We prove this by induction. At the beginning of the algorithm, the empty allocation is EF-$1$ to the allocator, and we assume at the beginning of the $\ell$-th round, the allocation is EF-$1$. 
We now prove after the $\ell$-th round, the allocation is still EF-$1$.
Let $x^{(k)}$ represent the item added into bundle $X$ in the $k$-th round.

Consider two arbitrary bundles $X, Y$, and assume two items added to these bundles are $g_i$ and $g_j$ respectively.
Without loss of generality, we assume $u(g_i)\ge u(g_j)$.
Suppose the two bundles are updated to $X', Y'$ after running the $\ell$-th round.
For bundle $X$, since there exists an item $g\in Y$ such that $u(X)\ge u(Y\setminus\{g\})$, we have $$u(X')=u(X\cup\{g_i\})\ge u(Y\cup\{g_j\}\setminus\{g\})=u(Y'\setminus\{g\}).$$
For bundle $Y$, because the items are sorted in descending order, we have $u\left(y^{(k-1)}\right)\ge u\left(x^{(k)}\right)$ for $2\le k\le \ell-1$, and $u\left(y^{(\ell-1)}\right)\ge u(g_i)$. 
Then we have 
$$u(Y')\geq\sum_{k=2}^\ell u\left(y^{(k-1)}\right) \ge \sum_{k=2}^{\ell-1} u\left(x^{(k)}\right)+u(g_i)=u\left(X'\setminus \{x^{(1)}\}\right).$$
Hence, the allocation after the $\ell$-th round is still EF-$1$ to the allocator.
\end{proof}

We conclude from the above two lemmas that the output allocation is doubly EF-$1$. 
Moreover, sorting the items takes $O(m\log m)$ time.
To allocate each group of items, finding a topological order of $G$ costs $O(n^2)$, and allocating one item and updating the envy-graph cost $O(n)$. 
The cycle-elimination algorithm takes $O(n^2)$ time to find a cycle and runs for at most $O(n^2)$ iterations because at least one edge is eliminated in each iteration. 
This process repeats for $\ceil{m/n}$ rounds.
The overall complexity of Algorithm~\ref{alg:identical_2ef1} is $O\left(m\left(\log m+n^3\right)\right)$.
Hence, Theorem~\ref{thm:identical_2ef1} holds.

\subsection{Additive Valuations with Two Agents}
\label{sec:additive-two-agents}
We then show the existence of a doubly EF-$1$ allocation with two agents with additive valuations and how such an allocation can be computed in polynomial time.

\begin{theorem}\label{thm:double_ef1_additive}
    When $n=2$ and the valuations are additive, there always exists a doubly EF-$1$ allocation, and such an allocation can be found in polynomial time.
\end{theorem}
\begin{proof}
    We assume $m$ is a multiple of $4$. Otherwise, we can add dummy items with value $0$ under the four utility functions $v_1,v_2,u_1,u_2$.
    Let $m=2k$. We will construct in polynomial time $(k+1)$ allocations $(A_1^0,A_2^0),(A_1^1,A_2^1),\ldots,(A_1^k,A_2^k)$ such that:
    \begin{enumerate}
        \item all the $(k+1)$ allocations are EF-$1$ with respect to $v_1$ and $u_1$;
        \item $|A_1^i|=|A_2^i|=k$ for each $i=0,1,\ldots,k$;
        \item $|A_1^i\cap A_1^{i+1}|=1$ (or, equivalently, $|A_1^i\cup A_1^{i+1}|=m-1$) for each $i=0,1,\ldots,k-1$, and $A_1^0\cap A_1^k=\emptyset$ (or, equivalently, $A_1^0=[m]\setminus A_1^k$).
    \end{enumerate}

    We first prove that, if such $k+1$ allocations can be found, then at least one of them is EF-$1$ with respect to both $v_2$ and $u_2$ (this will conclude the theorem given Property 1 above).
    We will make use of Proposition~\ref{prop:ef1_intersection}.
    If we consider a graph where the vertices are the allocations and there is an edge between two allocations $(A_1,A_2)$ and $(A_1',A_2')$ if and only if $|A_1\cap A_1'|\leq 1$, Proposition~\ref{prop:ef1_intersection} implies that the set of all allocations that are not EF-$1$ with respect to $v_2$ (or $u_2$) must form an independent set.
    Since two independent sets cannot cover all vertices in an odd cycle, and the $(k+1)$ allocations defined at the beginning form an odd cycle (recall that we have assumed $m$ is a multiple of $4$ and $k=m/2$), we conclude that at least one of the $k+1$ allocations is EF-$1$ with respect to both $v_2$ and $u_2$.
    
    It then remains to construct these $k+1$ allocations based on $v_1$ and $u_1$.

    We first sort the items by descending values based on $v_1$ and group every two items in the sorted list.
    We have a total of $k$ groups, and the $i$-th group contains the items with the $(2i-1)$-th largest value and the $(2i)$-th largest value respectively.
    Now we sort the groups and name the items in each group based on $u_1$.
    For each group of two items, we use the letter $a$ for the item with the larger value (according to $u_1$) and the letter $b$ for the other item.
    We sort the groups by descending values of $u_1(a)-u_1(b)$, and denote the $k$ groups by
    $$G_1=\{a_1,b_1\},G_2=\{a_2,b_2\},\ldots,G_k=\{a_k,b_k\}$$
    such that $u_1(a_1)-u_1(b_1)\geq u_1(a_2)-u_1(b_2)\geq \cdots\geq u_1(a_k)-u_1(b_k)$ and $u_1(a_i)\geq u_1(b_i)$ for each $i=1,\ldots,k$.

    We make the following observation for $v_1$.
    \begin{proposition}
        If $|A_1\cap G_i|=1$ for each $i=1,\ldots,k$, then $(A_1,A_2)$ is EF-$1$ with respect to $v_1$.
    \end{proposition}
    \begin{proof}
        Let $(H_1,\ldots,H_k)$ be the $k$ groups sorted by the descending value of $v_1$, where $H_i$ contains the items with the $(2i-1)$-th and the $(2i)$-th largest values.
        Note that $(H_1,\ldots,H_k)$ is a permutation of $(G_1,\ldots,G_k)$.
        Since $v_1(A_1\cap H_i)\geq v_1(A_2\cap H_{i+1})$ holds for $i=1,\ldots,k-1$, removing the item in $A_2\cap H_1$ from the bundle $A_2$ will make sure $A_1$ is weakly preferred under $v_1$.
    \end{proof}

    We will make sure the $k+1$ allocations satisfy the pre-condition of the proposition above, so the EF-$1$ property for $v_1$ is guaranteed.

    Next, we label each group $G_i$ by either ``$+$'' or ``$-$'', where $+$ indicates $a_i\in A_1$ and $-$ indicates $b_i\in A_1$.
    Notice that any labeling of the $k$ groups defines an allocation $(A_1,A_2)$ that satisfies the EF-$1$ property for $v_1$.
    We then make the following observation for $u_1$.
    
    \begin{proposition}\label{prop:pmlabel}
        Given any allocation defined by $k$ labels, if the number of ``$-$''s is at most one more than the number of ``$+$''s in the set of groups $\{G_1,G_2,\ldots,G_i\}$ for each $i=1,\ldots,k$, then this allocation is EF-$1$ for $u_1$.
    \end{proposition}
    \begin{proof}
        Consider an arbitrary valid labeling which defines an allocation $(A_1,A_2)$.
        Let $T^+$ be the set of groups with the label ``$+$'' and $T^-$ be the set of groups with the label ``$-$''.
        We first prove that, after excluding at most one group $G_t$ (for some $t\in[k]$) from $T^-$, there exists an injective mapping $\pi:T^-\setminus\{G_t\}\to T^+$ such that $i<j$ whenever $\pi(G_j)=G_i$.
        This mapping can be constructed inductively.
        Suppose all except one group with the label $-$ in $\{G_1,\ldots,G_i\}$ are mapped.
        We now extend the mapping to the set $\{G_1,\ldots,G_i,G_{i+1}\}$.
        If the label of $G_{i+1}$ is $+$, we are done.
        Otherwise, we discuss two cases.
        
        If $i$ is even, then the number of ``$+$''s is at least the number of ``$-$''s in  $\{G_1,\ldots,G_i\}$ (otherwise, the number of ``$-$''s is at least two more than the number of ``$+$''s, which is invalid).
        If all the ``$-$''s in $\{G_1,\ldots,G_i\}$ are mapped, we can leave $G_{i+1}$ unmapped. If one of the ``$-$'' in $\{G_1,\ldots,G_i\}$ is unmapped, at least one ``$+$'' group is not matched, and we can map $G_{i+1}$ to this group.

        If $i$ is odd, since we have assumed $G_{i+1}$ is labeled with ``$-$'', the number of `$+$''s is at least one more than the number of ``$-$''s in  $\{G_1,\ldots,G_i\}$ (otherwise, the number of ``$-$''s is at least two more than the number of ``$+$''s in $\{G_1,\ldots,G_i,G_{i+1}\}$, which is invalid).
        Therefore, at least one group with the label ``$+$'' in $\{G_1,\ldots,G_i\}$ is not matched with any group with the label ``$-$'' in $\{G_1,\ldots,G_i\}$, and we can map $G_{i+1}$ to this group with the label ``$+$''.

        Finally, we show that such a mapping can guarantee EF-$1$ for $u_1$.
        Let $G_t=(a_t,b_t)$ be the group with the label ``$-$'' that is not mapped.
        The allocation $(A_1\setminus\{b_t\},A_2\setminus\{a_t\})$ is envy-free due to the mapping, as $u_1(a_i)-u_1(b_i)\geq u_1(a_j)-u_1(b_j)$ for $i<j$ (so the ``envy'' created by the ``$-$'' groups are compensated by the ``advantage'' created by the mapped ``$+$'' groups).
        Since $u_1(A_1)\geq u_1(A_1\setminus\{b_t\})\geq u_1(A_2\setminus\{a_t\})$, the allocation $(A_1,A_2)$ is EF-$1$ for $u_1$.
    \end{proof}

    Finally, we are ready to construct the $k+1$ allocations.
    We will use a label sequence to describe an allocation.
    Let $(A_1^0,A_2^0)$ be defined by the alternating sequence $+-+-\cdots+-$.
    We define the remaining $k$ allocations inductively.
    Given $(A_1^i,A_2^i)$, the allocation $(A_1^{i+1},A_2^{i+1})$ is defined by flipping the labels of all the groups except for the $(i+1)$-th group.

    For example, suppose $m=8$ and the four groups are $G_1=\{a_1,b_1\},G_2=\{a_2,b_2\},G_3=\{a_3,b_3\},G_4=\{a_4,b_4\}$.
    The five allocations $(A_1^0,A_2^0),(A_1^1,A_2^1),\ldots,(A_1^4,A_2^4)$ are labeled by
    $$+-+-, ++-+, -++-, +-++, -+-+$$
    respectively, and $A_1^0,A_1^1,\ldots,A_1^4$ are
    $$\{a_1,b_2,a_3,b_4\},\{a_1,a_2,b_3,a_4\},\{b_1,a_2,a_3,b_4\},\{a_1,b_2,a_3,a_4\},\{b_1,a_2,b_3,a_4\}$$
    respectively.

    It is clear that the pre-condition for Proposition~\ref{prop:pmlabel} holds for all the $k+1$ allocations:
    for $(A_1^0,A_2^0)$, the labels are alternating; for each $(A_1^i,A_2^i)$, the $i$-th label is always ``$+$'', and the labels in both segments $G_1,\ldots,G_{i-1}$ and $G_{i+1},\ldots,G_k$ are alternating.
    Therefore, all allocations are EF-$1$ with respect to $u_1$.

    Lastly, $|A_1^i|=k$ and $|A_1^i\cap A_1^{i+1}|=1$ hold by our construction, and $A_1^0\cap A_1^k=\emptyset$ holds since the label of each group is flipped for $(k-1)$ times when transforming from $A_1^0$ to $A_1^k$ and $(k-1)$ is an odd number.

    The construction of the $k+1$ groups can clearly be done in polynomial time.
    Checking the EF-$1$ property for the $k+1$ groups for $v_2$ and $u_2$ can also be done in polynomial time.
    Thus, we conclude that a doubly EF-$1$ allocation can be found in polynomial time.
\end{proof} 

\begin{remark}
    Theorem~\ref{thm:double_ef1_additive} can be generalized to the case where $v_1,v_2$ (resp., $u_1,u_2$) are additive, and $u_1,u_2$ (resp., $v_1,v_2$) are monotone beyond additive, as the argument in Proposition~\ref{prop:ef1_intersection} do not require additivity.
    Note that if $u_1,u_2$ (resp., $v_1,v_2$) are monotone, to guarantee the allocation can be found in polynomial time, it is required that $u_i(\cdot)$ (resp., $v_i(\cdot)$) where $i\in[2]$ can be calculated in polynomial time.
\end{remark}

\subsection{Additive Valuations with General Number of Agents}\label{sec:2prop_logn_general_val}

This section studies double fairness for a general number of agents and additive valuations.

\begin{restatable}{theorem}{ThmDoubleLogn}
\label{thm:double_logn}
    For any $n\ge 2$, there always exists a doubly PROP-$(2 \ceil{\log n})$ allocation that can be computed in polynomial time. 
    In particular, when $n= 2^k$, there always exists a doubly PROP-$k$ allocation.
\end{restatable}

Before going into details of the proof, we first outline the high-level ideas.
We refer to the idea of Even-Paz algorithm~\citep{EVEN1984285}. 
Given $n$ agents, we first partition the agent set into two groups and try to allocate one bundle for each group.
After that, we fix the two bundles to the two groups and then do further allocating within groups recursively.
To guarantee the property of proportionality, we ensure that each agent and the allocator regard the ratio of the value the agent's group receives as about $1/2$. 
We first introduce the following notions used in our proof.

\begin{definition}
Let $L(v, t, S)$ be the sum of the largest $\min\{t, |S|\}$ values of items under the valuation $v$ among a given item set $S$.
Given a partition $(N_1, N_2)$ of agents and two positive integer $k_1, k_2$. We say that $(X_1, X_2)$ is a \emph{$2$-balanced PROP-$(k_1, k_2)$ allocation with respect to $(N_1, N_2)$} if 
$v(X_1) \ge \frac{\abs{N_1}}{n} v(M) - L(v, k_1, X_2)$ for each $v_i, u_i$ with $i\in N_1$ and $v(X_2) \ge \frac{\abs{N_2}}{n} v(M) - L(v, k_2, X_1)$ for each $v_i, u_i$ with $i\in N_2$.
\end{definition}

We give two lemmas on the existence of $2$-balanced PROP-$(k_1, k_2)$ allocations.
In the first lemma, we show that, for $n$ being an even number, there always exists a $2$-balanced PROP-$(\frac{n}2, \frac{n}2)$ allocation via Kneser graph, which is adopted to show the case where $n=2^k$. 
In the second lemma, we provide a constructive proof of how to find a $2$-balanced PROP-$(n-1, n)$ allocation via linear programming, and it does not put any restriction on $n$.

\begin{lemma}\label{lem:2_balanced_propn_2}
If $n$ is even, then for any $2$-partition  $(N_1, N_2)$ such that $\abs{N_1} = \abs{N_2} = n/2$, there always exists a $2$-balanced PROP-$\left(n/2, n/2\right)$ allocation. 
\end{lemma}
\begin{proof}
Let $n= 2s$.  
Denote by $\Pi$ the whole space of bundles with size $m/2$.
For each agent $i\in N_1$, we enumerate the bundles $X$ such that she does not regard as proportional, even if removing the $s$ largest items from the $M\setminus X$.
Formally, $\V_i =\{ X \in \Pi: v_i(X) < v_i(M)/2 - L(v_i, s, M\setminus X)\}$.
Symmetrically, for each agent $i\in N_2$, we enumerate the bundles that she considers proportional when taking the largest $s$ items from the remaining items: $\V_i = \{ X \in \Pi: v_i(X) > v_i(M)/2 - L(v_i, s, X)\}$.
Define $\U_i$ for each agent $i\in N$ by replacing $v_i$ by $u_i$ in the above formulas.  

\begin{proposition}\label{prop:N_1_x_1_x2_inter}
For each agent $i\in N_1$, $\abs{X_1 \cap X_2} > n$ for any $X_1, X_2 \in \V_i$.
\end{proposition}
\begin{proof}
We prove it by contradiction.
Assume $\abs{X_1 \cap X_2} \le n$. 
Let $B = M \setminus (X_1\cup X_2)$. Hence, $\abs{B} \le m - 2\cdot m/2  + n \le n$. According to the definition of $\V_i$, we have 
\begin{gather*}
v_i(X_1) < \frac{v_i(M)}2 - L(v_i, s, M\setminus X_1), \quad 
v_i(X_2) < \frac{v_i(M)}2 - L(v_i, s, M\setminus X_2)    
\end{gather*}
Sum them up, and we have
\begin{align*}
v_i(X_1) + v_i(X_2) & < v(M) -L(v_i, s, M\setminus X_1) - L(v_i, s, M\setminus X_2) \\
& = v_i\left(X_1 \cup X_2\right) + v_i(B) - L\left(v_i, s, M\setminus X_1\right) - L\left(v_i, s, M\setminus X_2\right)    
\end{align*}
Since $B = M \setminus (X_1\cup X_2) \subseteq M\setminus X_1$, then $L(v_i, s, B) \le L(v_i, s, M\setminus X_1)$. 
For the same reason, $L(v_i, s, B) \le L(v_i, s, M\setminus X_2)$.
Therefore, we can see that
$$
v_i(X_1) + v_i(X_2) < v_i\left(X_1 \cup X_2\right) + v_i(B) - L\left(v_i, s, B\right) - L\left(v_i, s, B\right) 
<  v_i\left(X_1 \cup X_2\right),
$$
which is impossible and completes the proof. 
\end{proof}

\begin{proposition}\label{prop:N_2_x_1_x2_inter}
For each agent $i\in N_2$, $\abs{X_1 \cap X_2} > n$ for any $X_1, X_2 \in \V_i$.
\end{proposition}
\begin{proof}
If not, assume $\abs{X_1 \cap X_2} \le n = 2^k$.
By the definition of $\V_i$ and the additivity of $v_i$, we have
\begin{align*}
v_i(X_1 \cup X_2)  & = v_i(X_1) + v_i(X_2) - v_i(X_1 \cap X_2) \\
& \ge  v_i(X_1) + v_i(X_2) - 2 \cdot L(v_i, 2^{k-1}, X_1\cap X_2) \\
& > \frac{v_i(M)}2 +  L(v_i, 2^{k-1}, X_1) + \frac{v_i(M)}2 +  L(v_i, 2^{k-1}, X_2) - 2 \cdot L(v_i, 2^{k-1}, X_1\cap X_2) \\ 
& \ge v_i(M) \ge v_i(X_1 \cup X_2)
\end{align*}
which leads to contradiction and completes the proof.
\end{proof}
Using almost the same proof, we have the same conclusions for $\U_i$. 
Thereafter, consider the Kneser graph $\mathcal{H} = \K(m, m/2, n )$. 
According to \Cref{lem:chromatic_for_kneser}, the choromatic number of $\mathcal{H}$ satisfies that $\chi(\mathcal{H}) \ge m - 2\cdot m/2 + 2n + 2 = 2n+2 > 2n$.
As we have already proved in \Cref{prop:N_1_x_1_x2_inter}, each of $\V_i$ and $\U_i$ does not contain two adjacent vertices of $\mathcal{H}$ and is thus an independent set. 
Since the number of these sets, $2n$, is less than  $\chi(\mathcal{H})$, the union of these $2n$ families of sets cannot cover the entire space -- all the $m/2$-subsets of $M$. 
Therefore, there exists a $m/2$-subset $X_0$ not belonging to any of $\V_1, \ldots, \V_n, \U_1, \ldots, \U_n$ and it is clear that $(X_0, M\setminus X_0)$ is a $2$-balanced PROP-$(n/2, n/2)$ allocation.     
\end{proof}

\begin{lemma}
\label{lem:2_balanced_propn}
For any $2$-partition  $(N_1, N_2)$ such that $\abs{N_1} = \floor{n/2}, \abs{N_2} = \ceil{n/2}$, there always exists a $2$-balanced PROP-$(n-1, n)$ allocation which can be computed in polynomial time.
\end{lemma}
\begin{proof}
Use variable $x_j$ to indicate the fraction of item $g_j$ allocated to group $N_1$ for every $g_j\in M$.   
Then consider the following linear program:
\begin{equation*}
\begin{array}{lc@{}ll}
\max  & \sum_j v_1(g_j)\cdot x_j - \frac{\lfloor n/2 \rfloor}{n} \cdot v_1(M) & & \\
\rm{s.t}  & \sum_j u_i(g_j) \cdot x_j \geq  \frac{\lfloor n/2 \rfloor}{n} \cdot u_i(M),  &&i \in N_1 \\
     & \sum_jv_i(g_j) \cdot x_j \geq  \frac{\lfloor n/2 \rfloor}{n} \cdot v_i(M),  &&i \in N_1\setminus \{1\} \\
     & \sum_ju_i(g_j) \cdot x_j \leq  \frac{\lfloor n/2 \rfloor}{n} \cdot u_i(M),  &&i \in N_2 \\
     & \sum_jv_i(g_j) \cdot x_j \leq  \frac{\lfloor n/2 \rfloor}{n} \cdot v_i(M),  &&i \in N_2 \\
     
     & 0\le x_j \le 1, &&j=1 ,\dots, m.
\end{array}
\end{equation*}
It is clear that $\bm{x^0} = \frac{\lfloor n/2 \rfloor}{n}\cdot \mathbf{1} $ is a feasible solution.
Since the objective function's value of $\bm{x^0}$ is $0$, the optimum of the linear program is non-negative. 
Notice that there also exists an optimal solution at a vertex of $\Omega$.
Denote it by $\bm{x}^* = (x_1^*, x_2^*, \ldots, x_m^*)$.

Since $\bm{x}^*$ is a vertex, according to the definition of vertex, there are at least $m$ constraints that are tight at $\bm{x}^*$.
Since there are totally $2n-1 + m$ constraints, at least $m - (2n-1)$ of the last $m$ constraints are tight. 
In other words, at least $m - (2n-1)$ of $x_1^*, \ldots, x_m^*$ are binary ($0$ or $1$). 
Without loss of generality, assume the first $t$ variables $x_1^*, \ldots x_t^*$ are between $0$ and $1$ and $1\le t \le 2n-1$. 
Let $O_1$ and $O_2$ be $\{g_j \in M: x_j^* = 1\}$ and $\{g_j \in M: x_j^* = 0\}$. 
Then let $X_1$ and $X_2$ be defined as 
$$
X_1 := \left\{g_j \in M: j\le \ceil{t/2} \right\}\cup O_1,\quad  X_2:= M \setminus X_1\,.
$$ 
Next we prove that $(X_1, X_2)$ is a $2$-balanced PROP-$n$ allocation. 
First, we can observe that, for each agent $i\in N_1$, the utility of $X_1$ under $v_i$ can be lower bounded by 
\begin{align*}
&\sum_{j \le \ceil{t/2}}v_i(g_j) + \sum_{g_j\in O_1}  v_i(g_j) \\
\text{(as $x_j^* \in [0,1]$)}\quad  \ge &  \sum_{j \le \ceil{t/2}}v_i(g_j)\cdot x_j^*  + \sum_{g_j\in O_1} v_i(g_j)\cdot x_j^* \\
\text{($x_j^* = 0$ for $g_j\in O_2$)} \quad =&  \sum_{g_j\in M} v_i(g_j) \cdot x_j^* - \sum_{\ceil{t/2} < j \le t} v_i(g_j) x_j^*\\
\text{(by feasibility of the solution)}\quad  \ge & \frac{\lfloor n/2 \rfloor}{n} \cdot  v_i(M)  - L\left(v_i, \floor{t/2}, X_2\right), 
\end{align*}
which implies that $(X_1, X_2)$ meets the proportionality constraint for agents in $N_1$. 
Similarly, we can lower bound $v_i(X_2)$ for every agent in $N_2$ by  
\begin{align*}
 & \sum_{\ceil{t/2} < j \le t} v_i\left(g_j\right) + \sum_{g_j\in O_2} v_i\left(g_j\right) \\
\text{(as $x_j^*\in [0,1]$)} \quad \ge & \sum_{\ceil{t/2} < j \le t} v_i(g_j)\cdot (1- x_j^*) + \sum_{j > t} v_i(g_j)\cdot (1- x_j^*)\\
\text{(as $x_j^*\in [0,1]$)} \quad \ge & \sum_{j\in M} v_i(g_j) \cdot (1-x_j^*) - \sum_{j\le \ceil{t/2}} v_i(g_j) \\
= & v_i(M)  -  \sum_{j\in M} v_i(g_j)x_j^* - \sum_{j\le \ceil{t/2}} v_i(g_j) \\
\text{(by feasibility of the solution)} \quad \ge & v_i(M)\cdot \left(1-\frac{\lfloor n/2 \rfloor}{n}\right) - L\left(v_i, n, X_1\right) \\
=  & v_i(M)\cdot \frac{\ceil{n/2}}n   - L\left(v_i, n, X_1\right).
\end{align*}
Overall, $(X_1, X_2)$ satisfies the definition of $2$-balanced PROP-$(n-1, n)$ allocation.
Finally, by \Cref{lem.lpvsol}, $(X_1, X_2)$ can be computed in polynomial time.
\end{proof}

Based on the two lemmas, we are now ready to prove \Cref{thm:double_logn}.
\begin{proof}[Proof of~\Cref{thm:double_logn}]
We first prove the second part of~\Cref{thm:double_logn} when $n=2^k$ by induction on $k$. 
When $k= 0$, this theorem trivially holds. 
If this theorem holds for any $k \le k_1$, consider the case of $k = k_1 + 1$. 
We partition the agents into two groups $N_1 = \left\{1, \ldots, 2^{k_1} \right\}$ and $N_2 = \left\{2^{k_1+1},\ldots, n\right\}$.
According to \Cref{lem:2_balanced_propn_2}, there exists a $2$-balanced PROP-$(\frac{n}2, \frac{n}2)$ allocation $(X_1, X_2)$. 
Next, we consider allocating $X_1$ to the first group and $X_2$ to the second group.

By the induction hypothesis, for the agent set $N_1$ and item set $X_1$, there exists a doubly PROP-$k_1$ allocation $(A_1, A_2, \ldots, A_{2^k_1})$. 
Likewise, there also exists a doubly PROP-$k_1$ allocation $(A_{2^{k_1} +1}, \ldots, A_n)$ for $N_2$ and $X_2$. 
To combine these two local properties into a global one while preserving proportional fairness guarantees, we rely on the following proposition:
\begin{proposition}\label{prop:local_prop_global_prop}
For an agent $i$ and three given integers $n_1, k_1, k_2 \in \mathbb{Z}^+$, if there exists sets $X$ and $A_i\subseteq X$ such that $v_i(A_i) \ge v_i(X) / n_1 - L(v_i, k_1, X)$ and $v_i(X) \ge n_1/n \cdot v_i(M) - L\left(v_i, n_1 \cdot k_2, M\setminus X\right)$, then $v_i(A_i) \ge v_i(M)/n - L(v_i, k_1 + k_2, M)$.
\end{proposition}
\begin{proof} 
This proposition can be concluded by the following inequalities,
\begin{align*}
v_i(A_i) & \ge \frac{ v_i(X)}{n_1} - L(v_i, k_1, X) 
\ge \frac{1}{n_1} \left(\frac{n_1}{n}\cdot v_i(M) - L\left(v_i, n_1 \cdot k_2, M\setminus X\right) \right) - L(v_i, k_1, X)\\
& = \frac{v_i(M)}{n}  - \frac{1}{n_1}\cdot L\left(v_i, n_1 \cdot k_2, M\setminus X\right) - L(v_i, k_1, X) \ge \frac{v_i(M)}{n}- L(v_i, k_1 + k_2, M)\qedhere
\end{align*}
\end{proof}

Using~\Cref{prop:local_prop_global_prop}, we can verify that $v_i(A_i) \ge \frac{1}n v_i(M) - L(v_i, k_1 + 1, M)$. Thus, the proof of the induction step is complete, and we have proved the second part.

The proof of the first part of~\Cref{thm:double_logn} where $n$ is not necessarily $2^k$ is similar.
we first partition the $n$ agents into two groups $N_1 = \left\{1, \ldots, \floor{n/2} \right\}$ and $N_2 = \left\{ \floor{n/2} +1,\ldots, n \right\}$. According to Lemma~\ref{lem:2_balanced_propn}, there exists a $2$-balanced PROP-$(n-1, n)$ allocation $(X_1, X_2)$. 
By the same inductive arguments as before, there exist two allocations $\left(A_1, \ldots, A_{\floor{n/2}}\right)$ and $\left(A_{\floor{n/2} + 1}, \ldots, A_n\right)$, which are respectively PROP-$2\ceil{\log(\floor{n/2})}$ and  PROP-$2\ceil{\log(\ceil{n/2})}$ for $N_1, X_1$ and $N_2, X_2$. 

Since $n - 1\le 2\cdot \floor{n/2}$ and $n \le 2\cdot \ceil{n/2}$, by applying Proposition~\ref{prop:local_prop_global_prop}, we can verify the allocation  $\left(A_1, \ldots, A_n\right)$ is doubly PROP-$\left(2\ceil{\log(\ceil{n/2})} + 2\right)$. It is not hard to verify that $\ceil{\log(\ceil{n/2})} = \ceil{\log n} -1$. Therefore, this allocation is also doubly PROP-$\left(2\ceil{\log n}\right)$.
\end{proof}

\subsection{Personalized Bi-valued Valuations}
\label{sec:bi-valued}
As we have shown in Theorem~\ref{thm:double_logn}, for additive valuation, when $n\ge 2$, a doubly PROP-$O(\log n)$ allocation always exists.
In this section, we further consider another common setting with personalized bi-valued valuations, where we make use of the round-robin algorithm based on an arbitrary ordering of agents.
Recall that in the personalized bi-valued setting, for any item $g\in M$ and agent $i \in [n]$, we have $v_i(g) \in \{p_{i,v}, q_{i,v}\}$ and $u_i(g)\in\{p_{i,u}, q_{i,u}\}$ for $0 \le p_{i,v} < q_{i,v}$ and $0 \le p_{i,u} < q_{i,u}$.
We show that a doubly PROP-2 allocation can be found in polynomial time.
\begin{theorem}\label{thm:doubly_prop2_for_bivalued_utility}
    When $u_i, v_i$ are both personalized bi-valued utility functions, there always exists a doubly PROP-$2$ allocation for any $n\ge 2$, and it can be computed in polynomial time.
\end{theorem}

Our algorithm is shown in Algorithm~\ref{algo:variant_of_rr}, which is an adaptation of the round-robin algorithm.
We assign the set of unallocated items to agents in rounds, where each agent receives exactly one item based on different types of items.
For each agent $i\in [n]$, according to $u_i$ and $v_i$, we partition the item set $M$ into four subsets, as follows:
\begin{align*}
& \Sione = \{g\in M:\ v_i(g) = q_{i,v},\ u_i(g) = q_{i,u}\},\quad 
\Sitwo = \{g\in M:\ v_i(g) = q_{i,v},\ u_i(g) = p_{i,u}\},\\
& \Sithw = \{g\in M:\ v_i(g) = p_{i,v},\ u_i(g) = q_{i,u}\},\quad
\Sifr = \{g\in M:\ v_i(g) = p_{i,v},\ u_i(g) = p_{i,u}\}.
\end{align*}

Observe that the items in $\Sione$ are the most valuable to agent $i$ since both she and the allocator regard the value of each of them to have a high value ($q_{i,v}$ and $q_{i,u}$).
Then, when it is agent $i$'s turn to pick an item in the round-robin algorithm, if there exists an unallocated item in $\Sione$, we assign it to her. 
Otherwise, we consider a function $\kappa_i^{(j)}$ defined as follows, 
$$
\kappa_i^{(j)} = |A_i \cap \mathcal{S}_i^{j}| - \frac1n\cdot |P\cap  \mathcal{S}_i^{j}|, \quad \text{for }j \in [4].
$$
where $A_i$ is the current bundle allocated to agent $i$ and $P$ is the set of allocated items, i.e., $P = \cup_{i\in [n]} A_i$.
If there exist unallocated items in both of $\Sitwo$ and $\Sithw$, we allocate an arbitrary item $g$ from $\Sitwo$ to agent $i$ if $\kappa_i^{(2)} \le \kappa_i^{(3)}$, and from $\Sithw$ otherwise.
When one of the sets exhausts, we assign agent $i$ an arbitrary item $g$ from the other set.
If both sets exhaust, we assign agent $i$ an arbitrary item $g$ from $\Sifr$.

\begin{algorithm}[t]
\caption{Finding doubly PROP-2 allocations for personalized bi-valued utility functions}
\label{algo:variant_of_rr}
Let $P$ represent the set of allocated items and initialize $P \leftarrow \emptyset$\;
\While{$P \neq M$} {
\For{$i= 1,\ldots, n$} {
\If{$\Sione \setminus P \neq \emptyset$} {
Pick an arbitrary item $g$ from $\Sione\setminus P$ and allocate it to $A_i$\; \label{line:allocate_s1}
}\ElseIf{$\Sitwo \setminus P \neq \emptyset$ and $\Sithw \setminus P \neq \emptyset$ \label{line:s2_s3_nonempty}} {
\If{$\kappa_i^{(2)} \le \kappa_i^{(3)}$} {
Pick an arbitrary item $g$ from $\Sitwo\setminus P$ and allocate it to $A_i$\;
}\Else{
Pick an arbitrary item $g$ from $\Sithw\setminus P$ and allocate it to $A_i$\;
}
}\ElseIf{$\Sitwo \setminus P \neq \emptyset$ or $\Sithw \setminus P \neq \emptyset$}{
Pick an arbitrary item $g$ from the non-empty set and allocate it to $A_i$\;
\label{line:pick_item_from_just_one}
}\ElseIf{$\Sifr \setminus P \neq \emptyset$}{
Pick an arbitrary item $g$ from $\Sifr \setminus P$ and allocate it to $A_i$\label{line:allocate_s4}\;
}
Update $P$ by $P\leftarrow P\cup \{g\}$\; 
}
}
\Return{$\mathcal{A} = (A_1,\ldots, A_n)$}
\end{algorithm}

Before presenting the detailed proof of Theorem~\ref{thm:doubly_prop2_for_bivalued_utility}, we first introduce the following two lemmas.

\begin{lemma}\label{lemma:kappa_12_13_cant_be_smaller_than_neg_1_at_the_same_time}
For any $i\in [n]$, $\kappa_i^{(1)} + \kappa_i^{(2)}$ and $\kappa_i^{(1)} + 
 \kappa_i^{(3)}$ cannot be less than $-1$ at the same time.
\end{lemma}
\begin{proof}
For the sake of contradiction, we assume $\kappa_i^{(1)} + \kappa_i^{(2)}$ and $\kappa_i^{(1)} + \kappa_i^{(3)}$ are both less than $-1$ at some point of the algorithm.
Since agent $i$ is assigned an item at each round, we have $\kappa_i^{(1)} + \kappa_i^{(2)} + \kappa_i^{(3)} \ge -1$.
Hence, we have $\kappa_i^{(3)} > 0$, which implies $\kappa_i^{(1)} < -1$.
According to Line~\ref{line:allocate_s1}, agent $i$ will pick items from set $\Sione$ first.
Thus, $\kappa_i^{(1)} \ge -1$ during the algorithm, which causes a contradiction.
\end{proof}

\begin{lemma}\label{lem:kap1_add_kap2_ge_neg2}
For any $i\in [n], \kappa_i^{(1)} + \kappa_i^{(2)} \ge -2$ and $\kappa_i^{(1)} + \kappa_i^{(3)} \ge -2$ for any $i\in [n]$.
\end{lemma}
\begin{proof}
Initially, since $A_i = \emptyset$, $\kappa_i^{(1)} + \kappa_i^{(2)} = 0$.
When $\Sione\setminus P \neq \emptyset$, since agent $i$ is assigned an item every $n$ items, $\kappa_i^{(1)} + \kappa_i^{(2)} \ge -1$. 
When agent $i$ first receives an item from $\Sitwo$ and $\Sithw$, there are at most $n-1$ items assigned to other agents after $\Sione$ exhausts, hence, $\kappa_i^{(1)} + \kappa_i^{(2)}\ge -2$ and $\kappa_i^{(1)} + \kappa_i^{(3)} \ge -2$.

We next prove this lemma holds during executing Line~\ref{line:s2_s3_nonempty} to Line~\ref{line:pick_item_from_just_one} by induction.
Consider agent $i$ receives an item $g$ from $\Sitwo$ or $\Sithw$ at the $k$-th round and $\kappa_i^{(1)} + \kappa_i^{(2)}$ is assumed to be no less than $-2$ at this time. 
Then, we discuss the following two cases.
\begin{itemize}[leftmargin=*]
    \item If $g$ belongs to $\Sitwo$, then $\kappa_i^{(1)} + \kappa_i^{(2)}$ is updated to $\kappa_i^{(1)} + \kappa_i^{(2)} + \frac{n-1}n$ after allocating $g$.
    Before allocating the next item to agent $i$, there are at most $n-1$ items of $\Sitwo$ to be allocated. 
    Thus, $\kappa_i^{(1)} + \kappa_i^{(2)}$ remains to be no less than $-2$ in the next round.
    \item If $g$ belongs to $\Sithw$, then we argue $\kappa_i^{(1)} + \kappa_i^{(2)} \ge -1$.
    Otherwise, $\kappa_i^{(1)} + \kappa_i^{(2)}$ and $\kappa_i^{(1)} + \kappa_i^{(3)}$ are both less than $-1$, which contradicts to Lemma~\ref{lemma:kappa_12_13_cant_be_smaller_than_neg_1_at_the_same_time}. 
    Thus, after allocating $g$ and the next $n-1$ items, $\kappa_i^{(1)} + \kappa_i^{(2)}$ will be updated to $\kappa_i^{(1)} + \kappa_i^{(2)} - 1 \ge -2$.
\end{itemize}
Based on the above discussion, $\kappa_i^{(1)} + \kappa_i^{(2)}\ge -2$ is maintained before allocating items of $\Sifr$ to agent $i$. 
Finally, since Line~\ref{line:allocate_s4} will not change $\kappa_i^{(1)} + \kappa_i^{(2)}$, the lemma will still not be violated.
Using the same analysis, we can also obtain that $\kappa_i^{(1)} + \kappa_i^{(3)} \ge -2$ for any $i\in [n]$.
\end{proof}

\begin{proof}[Proof of Theorem~\ref{thm:doubly_prop2_for_bivalued_utility}]
According to Lemma~\ref{lem:kap1_add_kap2_ge_neg2}, we have
$$
\abs{A_i \cap \left(\Sione \cup \Sitwo\right)} \ge \frac{1}n \abs{\Sione \cup \Sitwo} -2\,.
$$
Next, we show that the allocation returned by Algorithm~\ref{algo:variant_of_rr} is PROP-2 to agent $i$ under valuation function $v_i$.
Since agent $i$ is assigned an item at each round, hence, $\abs{A_i} \ge \lceil\frac{m}n\rceil - 1$.
Denote $\left\lceil \frac{|\Sione \cup \Sitwo |}{n} \right\rceil - 2$ by $t$. 
Without loss of generality, let $g_1,\ldots, g_t$ be arbitrary $t$ items of $A_i \cap (\Sione \cup \Sitwo)$. 
Thus, $v_i(A_i)$ is at least
\begin{align*}
& v_i (\{g_1,\ldots, g_t\}) + v_i(A_i \setminus \{g_1,\ldots, g_t\}) \\
\text{($g_1,\dots,g_t\in A_i \cap (\Sione \cup \Sitwo)$)}\quad \ge & t\cdot q_{i,v} + \left(\frac{m}n -  1 -t \right)p_{i,v}\\
\text{(as $m= \abs{\Sione \cup \Sitwo} + \abs{\Sithw \cup \Sifr}$)}\quad \ge & \left(\frac1n \abs{\Sione \cup \Sitwo} - 2\right)q_{i,v} + \left(\frac1n \abs{\Sithw \cup \Sifr} + 1\right)p_{i,v}\\
=  & \frac{v_i(M)}n - 2q_{i,v} + p_{i,v},
\end{align*}
which implies that the allocation $\mathcal{A}$ is PROP-2 to agent $i$ under $v_i$.
For the same reason, $\mathcal{A}$ is also PROP-2 to agent $i$ under $u_i$.
Lastly, it is clear that the algorithm runs in polynomial time.
\end{proof}

\section{Allocator's Efficiency}
In this section, we consider the problem of maximizing the allocator's efficiency subject to EF-$c$ or PORP-$c$ constraint for the agents.
Throughout this section, we assume the valuations are additive.
Note that we no longer consider the special case as in Sect.~\ref{sect:identical_double_ef1} where the allocator's valuation is identical to all agents, i.e., $u_1=\cdots=u_n$, as the problem becomes trivial under such case (all allocations have the same \sw).

\subsection{Maximizing Allocator's Efficiency for Two Agents} \label{sec:mae_2_agents}
\begin{theorem} \label{thm:mae_2_ab_neg}
    The problem of maximizing \sw subject to EF-$c$ for two agents is NP-hard to approximate to within factor $2$ even when the allocator's utility functions are binary and $c=1$.
\end{theorem}
\begin{proof}
We will present a reduction from partition. Given a partition instance $S=\{e_1,\dots,e_m\}$ such that $\sum_{k=1}^m e_k=1$, we construct an instance shown in the tables below.

\begin{minipage}[c]{0.5\textwidth}
\centering
\begin{tabular}{c|ccc}
     &  $g_k$ ($1\le k\le m$) &  $g_{m+1}$ &  $g_{m+2}$ \\
    \hline
    $v_1$ & $e_k$ & $1$ & $0$ \\ 
    $v_2$ & $e_k$ & $0$ & $1$ \\ 
    \hline
\end{tabular}
\captionof{table}{Agents' Utility Functions}
\end{minipage}
\begin{minipage}[c]{0.5\textwidth}
\centering
\begin{tabular}{c|ccc}
     &  $g_k$ ($1\le k\le m$) &  $g_{m+1}$ &  $g_{m+2}$ \\
    \hline
    $u_1$ & $0$ & $0$ & $1$ \\ 
    $u_2$ & $0$ & $1$ & $0$ \\ 
    \hline
\end{tabular}
\captionof{table}{Allocator's Utility Functions}
\end{minipage}

We can observe $2$ is an upper bound of \sw.
If the partition instance is a yes-instance, $S$ can be partitioned into $S_1$ and $S_2$ such that $\sum_{e_k\in S_1}e_k=\sum_{e_k\in S_2}e_k=\frac{1}{2}$. 
The allocation $A_1=S_1\cup\{g_{m+2}\}, A_2=S_2\cup\{g_{m+1}\}$ satisfies EF-$1$, and the \sw is $2$.

If the partition instance is a no-instance, assume the \sw is still 2, then the allocation should be $A_1=S'_1\cup\{g_{m+2}\}, A_2=S'_2\cup\{g_{m+1}\}$, where $S'_1\cup S'_2=S$.
To make this allocation EF-$1$, we have $\sum_{e_k\in S'_1}\ge\sum_{e_k\in S'_2}$
for agent $1$ and $\sum_{e_k\in S'_1}\le\sum_{e_k\in S'_2}$ for agent $2$.
Then $\sum_{e_k\in S'_1}e_k=\sum_{e_k\in S'_2}e_k=\frac{1}{2}$, which leads to a contradiction.
So the \sw is at most $1$.

Thus, the inapproximation factor is $2$.
\end{proof}

\begin{theorem} \label{thm:mae_2_aa_pos}
    The problem of maximizing \sw subject to EF-$c$ for two agents has a polynomial time $2$-approximation algorithm when the agents' utility functions are arbitrary.
\end{theorem}

We first introduce our algorithm. 
We initialize two empty bundles $S_1$ and $S_2$, and sort the items according to agent $1$'s utility in descending order.
Assume the sorted items are $\{g_1,\dots, g_m\}$, and use $G_i (i\ge 1)$ to denote a group of two items $\{g_{2i-1},g_{2i}\}$.
For each group $G_i (i\ge 1)$, we allocate one item to each bundle.
In particular, without loss of generality, we assume $v_2(S_1)\ge v_2(S_2)$ before allocating group $G_i$.
Then, if $v_2(g_{2i-1})\ge v_2(g_{2i})$, we allocate $g_{2i-1}$ to $S_2$ and $g_{2i}$ to $S_1$.
Otherwise, we allocate $g_{2i-1}$ to $S_1$ and $g_{2i}$ to $S_2$.
Notice that, in this algorithm, agent $1$'s utility function is used exclusively for the ordering of the item, and agent $2$'s utility function is used exclusively for deciding the allocation of the two items in each group.

After all the items are allocated, we consider the two allocations $(S_1,S_2)$ and $(S_2,S_1)$, and output the allocation with a higher allocator's efficiency.

\begin{proof}
First, we show both $(S_1,S_2)$ and $(S_2,S_1)$ satisfy EF-$1$, and thus satisfy EF-$c$.
For agent $1$, by taking $u=v_1$, the same arguments in the proof of Lemma~\ref{lem:identical_allocator_ef1} show that the allocation is EF-$1$ no matter which of $S_1$ or $S_2$ she takes.

We prove the allocations are EF-$1$ to agent $2$ by induction. 
When $S_1$ and $S_2$ are empty, both $(S_1,S_2)$ and $(S_2,S_1)$ are trivially EF-$1$.
Assume before allocating group $G_i$, both allocations satisfy EF-$1$.
Without loss of generality, assume $v_2(S_1)\ge v_2(S_2)$ and $v_2(g_{2i-1})\ge v_2(g_{2i})$.
By our algorithm, after allocating $G_i$, $S'_1=S_1\cup \{g_{2i}\}$ and $S'_2=S_2\cup \{g_{2i-1}\}$.
If agent $2$ receives $S'_1$, we have 
$$v_2(S'_1)\ge v_2(S_1)\ge v_2(S_2)=v_2\left(S'_2\setminus\{g_{2i-1}\}\right).$$ 
If agent $2$ receives $S'_2$, since there exists an item $g\in S_1$ such that $v_2(S_2)\ge v_2(S_1\setminus\{g\})$, we have $$v_2(S'_2)=v_2(S_2)+v_2(g_{2i-1})\ge v_2\left(S_1\setminus\{g\}\right)+v_2(g_{2i})=v_2(S'_1\setminus\{g\}).$$
Hence, both allocations are EF-$1$.

We next show the allocation with a higher \sw is a $2$-approximation to the optimal \sw. 
Without loss of generality, assume $\SW((S_1,S_2))\ge\SW((S_2,S_1))$.
Denote the optimal \sw by $\SW_{\OPT}$, and we have 
\begin{gather*}
    \SW_{\OPT}\le u_1(M)+u_2(M)=u_1(S_1)+u_2(S_2)+u_1(S_2)+u_2(S_1)\\
    =\SW((S_1,S_2))+\SW((S_2,S_1)).
\end{gather*}
Then, we have $\SW((S_1,S_2))\ge\frac{1}{2}\SW_{\OPT}$.

Since the algorithm outputs the allocation $(S_1,S_2)$, Theorem~\ref{thm:mae_2_aa_pos} holds.
\end{proof}

\subsection{Maximizing Allocator's Efficiency for Constant Number of Agents} \label{sec:mae_const_agents}
\begin{theorem} \label{thm:mae_const_aa_neg}
    The problem of maximizing \sw subject to EF-$c$ for any fixed $n\ge 3$ is NP-hard to approximate to within any factor that is smaller than $\floor{\frac{1+\sqrt{4n-3}}{2}}$ even when the allocator's utility functions are binary and $c=1$.
\end{theorem}
\begin{proof}
We adopt the reduction from partition by \citet{bu2022complexity}.
In the origin reduction, the key point is that there exists a super agent, and the social welfare almost depends exclusively on the super agent.
In addition, the super agent's utility functions are binary in Bu et al.'s reduction.

In our problem, we maintain the construction in the origin reduction and add an allocator.
For the super agent, we set the allocator's utility functions the same as the super agent's utility functions,
For other agents, we set the allocator's utility functions to be 0.
Hence, the \sw is equivalent to the social welfare in the origin reduction, so we get the same inapproximation result.
\end{proof}

\begin{theorem} \label{thm:mae_const_ba_pos}
    The problem of maximizing \sw subject to EF-$c$ for any fixed $n\ge 3$ can be found in polynomial time when the agents' utility functions are binary.
\end{theorem}
\begin{proof}
We can adapt the proof of Theorem 7.5 in~\citet{aziz2020computing}.
In their paper, they used the state of the form $(k,(t_{ij})_{i\neq j};(b_{ij})_{i\neq j})$ to state whether there exists such an allocation $\A=(A_1,\ldots,A_n)$ of $(g_1,\ldots,g_k)$ that $v_i(A_i)-v_i(A_j)=t_{ij}$ holds for every two agents $i,j\in N$ and item $g_{b_{ij}}$ is the item which maximizes $v_i$ in agent $j$'s bundle.

The difference from their algorithm is that our state of the form $(k,(t_{ij})_{i\neq j};(b_{ij})_{i\neq j})$ stores not only the information of the existence, but also the largest value of $\sum_{i\in[n]} u_i(A_i)$ for all satisfying allocations.
Besides, we do not need the information $b_{ij}$ as the utility functions are binary.
In particular, an allocation is EF-$c$ if and only if $v_i(A_i)-v_i(A_j)\leq t_{ij}$.

We can see that, for the state of the form  $(k,(t_{ij})_{i\neq j})$, the values of these parameters can determine the feasibility of the following allocation, and if we keep finding the largest \sw among all (partial) allocations stored in this state, we can find the EF-$c$ allocation with the largest \sw at the end.
\end{proof}

\subsection{Maximizing Allocator's Efficiency for General Number of Agents} \label{sec:mae_gen_agents}
For a general number of agents, we first consider EF-$1$, and show a strong inapproximation result even if both the agents' and the allocator's utility functions are binary.
\begin{theorem} \label{thm:mae_gen_bb_neg}
    For any $\epsilon>0$, the problem of maximizing \sw subject to EF-$c$ is NP-hard to approximate within factor $m^{1-\epsilon}$ or $n^{1/2-\epsilon}$, even if both the agents' and the allocator's utility functions are binary and $c=1$.
\end{theorem}
\begin{proof}
We adopt a similar idea as in~\cite{10.5555/3306127.3331927} and present a reduction from the maximum independent set problem. 
For a maximum independent set instance $G=(V,E)$ where $|V|=m$ and $|E|=n$, we construct the following maximizing \sw instance with $m$ items, $n+1$ agents, and an allocator. 
For each vertex $v\in V$, we construct an item $g_v$.
For each edge $e=(u,v)\in E$, we construct a normal agent $a_e$, whose values to her adjacent items $g_u$ and $g_v$ are $1$, and $0$ for other items.
Moreover, we construct a super agent $a_0$, whose value to all the items is $0$.
For the allocator, her value is $1$ if an item is allocated to the super agent, and $0$ otherwise.
We show that the maximum \sw is $k$ if and only if the maximum independent set in $G$ is of size $k$.

If $G$ contains an independent set $\mathcal{I}$ of size $k$, the maximum \sw is at least $k$ by allocating the items that correspond to the vertices in the independent set to the super agent. 
For normal agents, we allocate at most one adjacent item to her.
This allocation is valid since $|E|>|V\setminus \mathcal{I}|$.
We now prove the allocation is EF-$1$.
For the super agent $a_0$, she will envy no one.
For an arbitrary normal agent $i$, she will not envy $a_0$ because $a_0$ receives at most one of her adjacent items.
She will not envy another normal agent either for the same reason.

If $G$ contains no independent set whose size is larger than $k$, the maximum \sw cannot exceed $k$.
Otherwise, there must exist an edge $e=(u,v)$ that both $g_u$ and $g_v$ are allocated to $a_0$, and $a_e$ will envy $a_0$ even if $g_u$ or $g_v$ is removed from $a_0$'s bundle.
So the allocation is not EF-$1$.

Since the maximum independent set problem is known to be NP-hard to approximate to within a factor of $n^{1-\varepsilon}$ and $m=O(n^2)$, Theorem~\ref{thm:mae_gen_bb_neg} holds.
\end{proof}

We show that a simple variant of the round-robin algorithm can achieve $m$-approximation for EF-$c$ allocations.
\begin{theorem} \label{thm:mae_gen_aa_pos}
    The problem of maximizing \sw subject to EF-$c$ has a $m$-approximation algorithm when both the agents' and the allocator's utility functions are arbitrary.
\end{theorem}
\begin{proof}
Let the allocator allocate a single item to a single agent with the highest value $u_i(g_j)$ for $1\le i\le n, 1\le j\le m$ to agent $i$.
Then the agents use the round-robin algorithm to allocate the remaining items, where agent $i$ receives an item at the end of each round.
The allocation is EF-$1$ (and is thus EF-$c$) guaranteed by the round-robin algorithm and is a trivial $m$-approximation to the optimal \sw.
\end{proof}

We note that it is an interesting open question whether there is an $O(n)$-approximation algorithm since the impossibility result of $m^{1-\varepsilon}$ as stated in Theorem~\ref{thm:mae_gen_bb_neg} occurs when $m < n$.
We next turn our attention to a weaker notion PROP-$c$. We first show that, even if the allocator's utility functions are binary, we can still get the following impossibility result.

\begin{theorem} \label{thm:maep_gen_ab_neg}
    The problem of maximizing \sw subject to PROP-$c$ is NP-hard to approximate within factor $2$  even if the allocator's utility functions are binary and $c=1$.
\end{theorem}
\begin{proof}
We will present a reduction from the partition problem.
For a partition instance $S =\{e_1, e_2, \ldots, e_m\}$, where $\sum_{k=1}^m e_k = x$, we construct an instance as follows.
Let $n = 2s$ be an even integer. 
The instance contains $n$ agents and $s \cdot m + n + 2$ items.
We first construct $s$ groups of items called \emph{partition items}, where each group contains $m$ items. 
Denote the items within $k$-th group by $g_1^{(k)}, \ldots, g_m^{(k)}$.
In addition, we also construct other $n+2$ items called \emph{pool items}. 
They are denoted by $g_{s + 1}, \ldots, g_{s + n+2}$. 

\begin{center}
    \begin{tabular}{ccccccccccc}
        \toprule
        item &  $g_k^{(1)}$ ($1\le k \le m$) &  $\cdots$ & $g_k^{(s)}$ & $g_{s+1}$ & $g_{s+2}$ &  $\cdots$  & $g_{s+n-1}$ & $g_{s+n}$ & $g_{s+n+1}$ & $g_{s+n+2}$\\
        \hline
        $v_{1}$ &   $e_k$ &  $0$ & $0$ & $0$ & $C$ &  $C$  & $C$  & $C$ & $C$ & $C$ \\
        $v_{2}$ &   $e_k$ &  $0$ & $0$ & $C$ & $0$ &  $C$  & $C$  & $C$ & $C$ & $C$ \\
        $\vdots$ &  $\vdots$   &  $\ddots$ & $\vdots$ & $\vdots$ &  $\vdots$ &  $\ddots$  & $\vdots$ & $\vdots$ & $\vdots$ & $\vdots$\\
        $v_{2s-1}$& $0$   &  $0$ & $e_k$ & $C$ & $C$ &  $C$  & $0$  & $C$ & $C$ & $C$\\
        $v_{2s}$ &  $0$   &  $0$ & $e_k$ & $C$ & $C$ &  $C$  & $C$  & $0$ & $C$ & $C$\\
        $u_{1}$ &   $0$   &  $0$ & $0$   & $1$ & $0$ &  $0$  & $0$  & $0$ & $0$ & $0$ \\
        $u_{2}$ &   $0$   &  $0$ & $0$   & $0$ & $1$ &  $0$  & $0$  & $0$ & $0$ & $0$ \\
        $\vdots$ &  $\vdots$ &  $\vdots$ & $\vdots$ & $\vdots$ & $\vdots$ &  $\ddots$  & $\vdots$  & $\vdots$ & $\vdots$ & $\vdots$ \\
        $u_{2s-1}$ &  $0$ &  $0$ & $0$ & $0$ & $0$ &  $\cdots$  & $1$  & $0$ & $0$ & $0$\\
        $u_{2s}$ &  $0$ &  $0$ & $0$ & $0$ & $0$ &  $\cdots$  & $0$  & $1$ & $0$ & $0$\\
        \bottomrule
    \end{tabular}
\end{center}

The utility functions of the $n$ agents and the allocator are defined in the above table.
For each partition items $g_k^{(i)}$, agent $2i-1$ and $2i$ have value $e_k$ while other agents have value $0$.
The allocator also has value $0$ no matter to whom it is allocated.
For each pool item $g_{s+k} (1\le k\le n)$, agent $k$ has value $0$ and other agents have value $C$, where $C=(\frac{n}{2}-1)\cdot x$.
The allocator has value $1$ if $g_{s+k}$ is allocated to agent $k$ and $0$ otherwise.
For item $g_{s+n+1}$ and $g_{s+n+2}$, all agents has value $C$ while the allocator has value $0$.

For a proportional allocation, each agent must receive a bundle with value at least $\frac{(n-1)\cdot x}{2}$.

We observe the upper bound of the \sw is $n$. If the partition instance is a yes-instance, we allocate $\{g_{k}^{(i)}\}_{1\leq k\leq m}$ to agent $2i-1$ and $2i$ such that each of the two agents receives a value of exactly $x$.
We allocate $g_{s+i} (1\le i\le n)$ to agent $i$ for each $i\in[n]$, and $g_{s+n+1}$ and $g_{s+n+2}$ to some arbitrary agents.
Each agent receives at least $\frac{x}{2}$.
This allocation is PROP-$1$ because, if each agent takes an extra item with value $C$, she will reach the proportional value.
The \sw is $n$ for this allocation.

If the partition instance is a no-instance, at least $s$ agents receive values that are less than $\frac{x}{2}$ from the partition items.
They need to take a pool item with value $C$ to be PROP-$1$, where the allocator's value is $0$ for it. 
Then, the \sw will be at most $n-s+2$.

Hence, the inapproximation factor is at most $2$.
\end{proof}

If the agents' utility functions are binary but not the allocator's, we can use linear programming to prove the following result.

\begin{theorem} \label{thm:maep_gen_ba_pos}
    When agents' utility functions are binary, the problem of maximizing \sw subject to PROP-$c$ can be solved exactly in polynomial time by linear programming.
\end{theorem}

\begin{proof}
    We first model this problem with a linear program. For each agent $i\in N$ and each item $g_j\in M$, we use one decision variable $x_{ij}$ to represent the fraction of item $g_j$ allocated to agent $i$. We can get the following linear program.
    \begin{equation*}
    \begin{array}{ll@{}rr}
        \textbf{maximize} & \displaystyle\sum_{i\in[n],j\in[m]} u_i(g_j)\cdot x_{ij} & & \\
        \textbf{subject to} & \displaystyle\sum_{j\in[m]} v_i(g_j)\cdot x_{ij}\geq \left\lceil \frac{1}{n}\sum_{j\in[m]} v_i(g_j) \right\rceil -c, & \forall i\in[n], & \qquad (a) \\
        &\displaystyle\sum_{i\in[n]} x_{ij}\leq 1, & \forall j\in[m], &\qquad (b) \\
        &x_{ij}\geq 0, &\forall i\in[n],j\in[m]. & \qquad (c)
    \end{array}
    \end{equation*}
Constraints (a) ensure the corresponding fractional allocation is PROP-$c$, and Constraints (b) ensure the feasibility of the allocation.

Since the feasible region for this linear programming is bounded by Constraints (b), and setting all $x_{ij}$ as $\frac{1}{n}$ is a feasible solution, there exists an optimal solution to this linear program.

With the integral constraint vector and applying Lemma~\ref{lem.tum} and Lemma~\ref{lem.lpvsol}, it suffices to show the coefficient matrix $\mathbf{A}$ for this linear programming is totally unimodular (TUM).
This follows straightforwardly from Lemma~\ref{lem:bipartite_utm}.
\end{proof}

\section{Conclusion and Future Work}
In this paper, we initialize the study of a new fair division model that incorporates the allocator's preference. 
We focused on the indivisible goods setting and mainly studied two research questions based on the allocator's preference:
1) How to find a doubly fair allocation? 2) What is the complexity of the problem of maximizing the allocator's efficiency subject to agents' fairness constraint?

We believe this new model is worth more future studies.
For example, could we extend our results to the setting with more general valuation functions, e.g., submodular valuations?
It is also an interesting (and challenging) problem to study what is the minimum number $c$ where a doubly EF-$c$/PROP-$c$ allocation is guaranteed to exist.
Indeed, we do not know any lower bound to $c$. In particular, we do not know if a doubly EF-$1$, or even doubly PROP-$1$, allocation exists even for binary utility functions.
We have searched for a non-existence counterexample with the aid of computer programs, and a non-existence counterexample seems hard to find.

On the other hand, our current techniques about the chromatic number and linear programming seem to have their limitations for further reducing the upper bound of $c$.
Our current definition of $\Gamma$-graph and current technique with the Kneser graph can only analyze a bi-partition of the items (in particular, a bi-partition with an equal size $m/2$ for the Kneser graph, which is crucial for Proposition~\ref{prop:N_1_x_1_x2_inter}).
For the analysis in the latter part of Theorem~\ref{thm:double_logn} with a general number of agents using the Kneser graph, the value of the bundle must be exactly \emph{half} of the total value up to the addition of $c$ items, which is why we need $n$ to be an integer power of $2$.
Moreover, the nature of the analysis based on the $\Gamma$-graph and Kneser graph makes the existence proof non-constructive.
Our linear programming technique in the former part of Theorem~\ref{thm:double_logn}, on the other hand, provides a weaker bound on $c$.
It seems to us that a Kneser graph captures more structural insights about our problem than a linear program.
Nevertheless, linear programming-based techniques provide a constructive existence proof.

It is fascinating to see how these techniques can be further exploited and if the above-mentioned limitations can be bypassed.
Unearthing new techniques for closing the gap between the upper bound and the lower bound of $c$ may also be necessary.

\subsection{Fair Division with Multiple Sets of Valuations}
\label{sec:discuss-multiple}
In our double fairness setting, we aim to find an allocation $(A_1,\ldots,A_n)$ that is fair with respect to \emph{two} valuation profiles $(u_1,\ldots,u_n)$ and $(v_1,\ldots,v_n)$, one for the agents and one for the allocator.
A natural generalization of this is to consider allocations that are fair with respect to $t$ valuation profiles for general $t$.
The problem of fair division with more than two sets of valuations is also well-motivated in many applications.
For example, there may be more than one ``allocator'' in many scenarios.
Taking the example of educational resource allocation in Sect.~\ref{sec:intro}, the government may consist of multiple parties, and it is desirable to find an allocation that is fair for all parties.
For another example, an agent's valuation of the items may be multi-dimensional.
When allocating employees to the departments of an organization, fairness is evaluated by multiple factors including employees' salaries, skill sets, diversity, etc.
When dealing with multiple sets of valuations, different fairness criteria can be considered.

As a natural generalization of the setting in this paper, we can consider allocations that are EF-$c$ or PROP-$c$ for \emph{all} valuation profiles.
This coincides with the setting of group fairness with group sizes satisfying $n_1=n_2=\cdots=n_k=t$ (see the last paragraph of Sect.~\ref{sect:relatedwork} for further discussions).
In contrast to our results in Theorem~\ref{thm:identical_2ef1} and Theorem~\ref{thm:double_ef1}, even when there are only two agents and the valuations of the agents in each of the three profiles $(u_1,u_2),(v_1,v_2),(w_1,w_2)$ are identical (i.e., $u_1=u_2$, $v_1=v_2$, and $w_1=w_2$) and binary, a triply EF-$1$ allocation may fail to exist.
In the example in Table~\ref{tab:counterexampletriple}, it is easy to see that Items 1 and 2 must not be in the same bundle based on $u_1$ and $u_2$, Items 1 and 3 must not be in the same bundle based on $v_1$ and $v_2$, and Items 2 and 3 must not be in the same bundle based on $w_1$ and $w_2$.
Clearly, no allocation satisfies these.
It is then natural to ask for which values of $c$ there is always an allocation that is EF-$c$ for all valuation profiles.

\begin{table}[h]
    \centering
    \begin{tabular}{cccc}
    \hline
         & Item 1 & Item 2 & Item 3 \\
    \hline
     Values based on $u_1=u_2$    & 1 & 1 & 0\\
     Values based on $v_1=v_2$    & 1 & 0 & 1\\
     Values based on $w_1=w_2$    & 0 & 1 & 1\\
     \hline
    \end{tabular}
    \caption{An example where a triply EF-$1$ allocation fails to exist.}
    \label{tab:counterexampletriple}
\end{table}

Another compelling criterion is to make the allocation fair with respect to $\ell$ out of $k$ valuation profiles.
Using the well-studied criterion EF-$1$ as an example, considering $k$ valuation profiles 
$$\left\{(u_1^{(1)},\ldots,u_n^{(1)}),(u_1^{(2)},\ldots,u_n^{(2)}),\ldots,(u_1^{(k)},\ldots,u_n^{(k)})\right\}$$ 
and given a parameter $\ell\leq k$, our goal is to find an allocation $(A_1,\ldots,A_n)$ such that, for each agent $i$, there exists $\ell$ valuation functions from $\{u_i^{(1)},\ldots,u_i^{(k)}\}$ such that the allocation satisfies EF-$1$ with respect to these $\ell$ valuation functions.
It is interesting to find out for which values of $k$ and $\ell$ this is possible.

\section*{Acknowledgments}
The research of Biaoshuai Tao was supported by the National Natural Science Foundation of China (Grant No. 62472271 
and 62102252).
The research of Shengxin Liu was partially supported by the National Natural Science Foundation of China (No. 62102117), by the Shenzhen Science and Technology Program (No. RCBS20210609103900003 and GXWD20231129111306002), and by the Guangdong Basic and Applied Basic Research Foundation (No. 2023A1515011188), and by CCF-Huawei Populus Grove Fund (No. CCF-HuaweiLK2022005).

\bibliographystyle{abbrvnat}
\bibliography{arxiv/reference}

\appendix
\section{Omitted Proofs for Section 3}
\label{app:sec3}
\DefiGood*
\propTriTetra*
\begin{proof}
We first consider $X_2\setminus X_1$ and $X_3\setminus X_1$.
Define the element $g$ as follows:

\begin{itemize}
    \item if both $X_2\setminus X_1$ and $X_3\setminus X_1$ have cardinality $2$, we must have $(X_2\setminus X_1)\cap (X_3\setminus X_1)\neq\emptyset$.
    Otherwise, the common neighbor $Y$ of $X_1$, $X_2$, and $X_3$ must intersect $X_2\setminus X_1$ at exactly one element (if the intersection is empty, then $|X_1\cup Y|\leq n-2$; if the intersection has size $2$, then $|X_2\cap Y|\geq 2$), and it must intersect $X_3\setminus X_1$ at exactly one element (for the same reason); in this case, we will have $|X_1\cup Y|\leq n-2$, leading to a contradiction.
    Then let $g$ be an arbitrary element in $(X_2\setminus X_1)\cap (X_3\setminus X_1)$;
    \item if $|X_2\setminus X_1|=2$ and $|X_3\setminus X_1|\leq1$, let $g$ be an element in $X_2\setminus X_1$, and choose $g$ to be in $X_3$ if $X_3\cap(X_2\setminus X_1)\neq\emptyset$; 
    \item if $|X_2\setminus X_1|\leq1$ and $|X_3\setminus X_1|=2$, let $g$ be an element in $X_3\setminus X_1$, and choose $g$ to be in $X_2$ if $X_2\cap(X_3\setminus X_1)\neq\emptyset$;
    \item otherwise, $g$ is considered ``undefined''. 
\end{itemize}
We next consider $X_1\setminus X_2$ and $X_1\setminus X_3$ and define element $h$ as follows:
\begin{itemize}
    \item if both $X_1\setminus X_2$ and $X_1\setminus X_3$ have cardinality $2$, for the similar reasons as above, we must have $(X_1\setminus X_2)\cap(X_1\setminus X_3)\neq\emptyset$.
    Let $h$ be an arbitrary element in $(X_1\setminus X_2)\cap (X_1\setminus X_3)$;
    \item if $|X_1\setminus X_2|=2$ and $|X_1\setminus X_3|\leq1$, let $h$ be an element in $X_1\setminus X_2$, and choose $h$ to avoid $X_3$ if $(X_1\setminus X_2)\not\subseteq X_3$;
    \item if $|X_1\setminus X_2|\leq1$ and $|X_1\setminus X_3|=2$, let $h$ be an element in $X_1\setminus X_3$, and choose $h$ to avoid $X_2$ if $(X_1\setminus X_3)\not\subseteq X_2$;
    \item otherwise, $h$ is considered ``undefined''.
\end{itemize}

Next, we construct the set $X_4$ as follows:
Let $X_4=(X_1\setminus\{h\})\cup\{g\}$ (replace $\{g\}$ or $\{h\}$ by the empty set if undefined).

Since $X_1X_2$ and $X_1X_3$ are both bad edges, at least one of $g$ and $h$ is defined.
Then $X_4\neq X_1$, and $X_1X_4$ is an edge in $\N$ by \Cref{prop:form_edge}. 
Moreover, it is a good edge by our definition.
For $X_2$, since $X_1X_2$ is a bad edge, at least one of $|X_2\setminus X_1|=2$ or $|X_1\setminus X_2|=2$ holds.
Therefore, compared with $X_4$, $X_2$ either contains at least one more element, the one in $(X_2\setminus X_1)\setminus\{g\}$, or at least one less element, the one in $(X_1\setminus X_2)\setminus\{h\}$ (or both).
Thus, $X_2\neq X_4$. 
Next, we prove $X_2X_4$ is an edge.

We first check that $|X_2\setminus X_4|\leq 1$: 
\begin{itemize}
    \item Firstly, we check that $h$ is never in $X_2$. The only possibility for $h\in X_2$ is when $|X_1\setminus X_2|\leq 1$, $|X_1\setminus X_3|=2$, and $(X_1\setminus X_3)\subseteq X_2$. We will derive a contradiction if this happens.
    Since $|X_1\setminus X_2|\leq 1$ and $X_1X_2$ is a bad edge, we must have $|X_2\setminus X_1|=2$.
    By $(X_1\setminus X_3)\subseteq X_2$ and $|X_1\setminus X_3|=2$, we must have $|X_2\setminus X_3|=2$ and $(X_2\setminus X_3)\subseteq X_1$.
    As a conclusion, we have $|X_2\setminus X_1|=2$, $|X_2\setminus X_3|=2$, and $(X_2\setminus X_1)\cap (X_2\setminus X_3)=\emptyset$.  
    However, this contradicts to that $X_1X_2X_3$ is a triangle in $\N$:
    the common neighbor $Y$ for $X_1,X_2,$ and $X_3$ must contain exactly one element in $X_2\setminus X_1$ (to ensure $|X_1\cup Y|\geq n-1$) and exactly one element in $X_2\setminus X_3$ (to ensure $|X_3\cup Y|\geq n-1$), but in this case $|X_2\cap Y|\geq 2$.
    \item With $h\notin X_2$, we have $X_2\setminus X_4=X_2\setminus (X_1\cup\{g\})$.
    If $|X_2\setminus X_1|=2$, then $X_2\setminus X_4=(X_2\setminus X_1)\setminus\{g\}$, which has size $1$; if $|X_2\setminus X_1|\leq 1$, then  $|X_2\setminus X_4|\leq |X_2\setminus X_1|=1$.
\end{itemize}
Finally, we check $|X_4\setminus X_2|\leq1$:
\begin{itemize}
    \item Firstly, we check that, if $g$ is defined, then $g\in X_2$.
    The only possibility when $g$ is defined and $g\notin X_2$ is when $|X_2\setminus X_1|\leq 1$, $|X_3\setminus X_1|=2$, and $X_2\cap (X_3\setminus X_1)=\emptyset$. 
    We will derive a contradiction if this happens.
    Since $X_1X_2$ is a bad edge and $|X_2\setminus X_1|\leq 1$, we have $|X_1\setminus X_2|=2$.
    By $X_2\cap (X_3\setminus X_1)=\emptyset$ and $|X_3\setminus X_1|=2$, we have $|X_3\setminus X_2|=2$ and $(X_1\setminus X_2)\cap (X_3\setminus X_2)=\emptyset$.
    Then, similar as before, the common neighbor $Y$ of $X_1,X_2$, and $X_3$ must intersect $X_1\setminus X_2$ at exactly one element and intersect $X_3\setminus X_2$ at exactly one element, but this would imply $|Y\cup X_2|\leq n-2$, leading to a contradiction.
    \item With $g\in X_2$, we have $X_4\setminus X_2=X_1\setminus\{h\}\setminus X_2$. If $|X_1\setminus X_2|=2$, then $h\in X_1\setminus X_2$ and $|X_4\setminus X_2|=1$; if $|X_1\setminus X_2|\leq1$, then $|X_4\setminus X_2|\leq |X_1\setminus X_2|\leq1$.
\end{itemize}
We have checked $X_2\neq X_4$, $|X_2\setminus X_4|\leq1$, and $|X_4\setminus X_2|\leq 1$.
Thus, $X_2X_4$ is a good edge.
For the exact same reason, $X_3X_4$ is a good edge.
\end{proof}

\section{Divisible Resources with Allocator's Preference}
In this section, we discuss the two research questions mentioned before Sect.~1.1 for divisible resources.
The two fairness notions envy-freeness and proportionality discussed in this paper can be satisfied exactly in the setting with divisible resources.
Thus, we focus on exact envy-freeness and proportionality here.
There are multiple different models for divisible resources.

\subsection{Divisible Homogeneous Items}
The simplest setting is the same setting as it is in Sect.~2 except that we now allow fractional allocations, i.e., each $g_j\in M$ can now be split among the agents.
Each item $g_j$ is assumed to be homogeneous: each agent $i$'s value on an $\alpha$-fraction of $g_j$ is given by $\alpha\cdot v_i(\{g_j\})$.

For the first research question, there exists a trivial doubly envy-free and doubly proportional allocation: just allocate each item evenly to the agents such that each agent gets a $1/n$ fraction of each item.

For the second research question, the problem of maximizing \sw subject to the envy-free/proportional constraint can be formulated by a linear program.
Let $x_{ij}$ be the fraction of item $j$ allocated to agent $i$.
It is straightforward to see that the envy-free constraints and the proportional constraints are linear in $x_{ij}$'s, and the \sw is also a linear expression of $x_{ij}$'s.
This gives us a polynomial time algorithm to solve this constrained optimization problem exactly.

\subsection{Cake Cutting}
Another well-studied model for divisible resources is the \emph{cake-cutting} model.
In the cake-cutting model, a single piece of heterogeneous resource, modeled by the interval $[0,1]$, is allocated to $n$ agents.
Each agent $i$ has a \emph{value density function} $f_i:[0,1]\to\mathbb{R}_{\geq 0}$, and her value of a subset $S\subseteq[0,1]$ is given by the Riemann integral
$$\int_Sf_i(x)dx.$$
The fairness notions envy-freeness and proportionality can then be defined accordingly.

From computer scientists' perspective, we then face the problem of succinct representation of each $f_i$.
There are two different approaches in the past literature.

\subsubsection{Piecewise-constant value density functions}
In the first approach, each $f_i$ is assumed to be piecewise-constant, where the interval $[0,1]$ can be partitioned into many subintervals where $f_i$ is a constant on each of them (see, e.g., \citep{cohler2011optimal,bei2012optimal,brams2012maxsum}).
Piecewise-constant functions can be succinctly represented and can approximate real functions with arbitrarily good precision.

This model then reduces to the previous model: we can find all the points of discontinuity of $f_1,\ldots,f_n$; this will partition $[0,1]$ into many subintervals where each $f_i$ is a constant on each of them, and each of these subintervals can be viewed as an ``item'' in the previous model.

Therefore, all results in the previous model apply here.
We can find a doubly envy-free (and thus doubly proportional) allocation in polynomial time, and we can solve the problem of maximizing the \sw subject to agents' envy-free/proportional constraints in polynomial time by linear programming.

\subsubsection{General value density functions}
If no assumption is made on the value density functions, the existence of a doubly envy-free allocation still holds.
\citet{alon1987splitting} showed that for $m$ agents and any positive number $n$, there exists an allocation $(A_1,\ldots,A_n)$ such that each $A_i$ has value exactly $\frac1n$ of the value of $[0,1]$ based on each agent's value density function.
By taking $m=2n$, this implies the existence of a doubly envy-free allocation.

The second problem of maximizing the \sw is related to the computational complexity, so we need to define a model to access the value density functions.
A commonly used one is the \emph{Robertson-Webb query model}~\citep{RobertsonWe98}.
However, finding an envy-free allocation under this model is already challenging and solved only recently, with an exponential time complexity~\citep{AzizMa16,AzizMa16-STOC}.

\end{document}